\RequirePackage{snapshot}
\documentclass[11pt,leqno]{article}
\usepackage{ltexpprt}

\usepackage{amsmath,amssymb,amsfonts}
\usepackage{mathtools}
\mathtoolsset{showonlyrefs=true}
\allowdisplaybreaks

\usepackage[english]{babel}
\usepackage{lmodern}
\usepackage[T1]{fontenc}

\usepackage{kpfonts}
\usepackage{xspace}
\usepackage{graphicx}
\usepackage{float}
\floatstyle{boxed}
\newfloat{pseudocode}{thb}{pseudo}
\floatname{pseudocode}{Pseudocode}

\usepackage{enumitem}
\setlist[description]{leftmargin=18pt}

\newtheorem{definition}{Definition}
\newtheorem{remark}{Remark}
    \newcommand{\qed}{\hspace{\stretch{1}$\square$}}
    \newcommand{\proofof}[1]{\noindent\textit{Proof of~#1. }}

    \global\long\def\pull{\mathcal{PULL}}
    \global\long\def\push{\mathcal{PUSH}}

    \newcommand{\consttwo}{N}

    \newcommand{\corelem}{Message Reduction Theorem}
    \newcommand{\corelemma}{the \corelem\xspace}
    \newcommand{\numbits}{\ell}
    \newcommand{\mesnum}{\eta}
    
    \newcommand{\finaliter}{\tau}
    \newcommand{\clock}{C}
    \newcommand{\clockname}{clock\xspace}
    \newcommand{\consconvprob}{2}
    \newcommand{\consconvtime}{\const \log n}
    \newcommand{\subphasescorelem}{ \frac{\mesnum}{2} }
    \newcommand{\trialbias}{\epsilon_{end}}

    \newcommand{\homemadeloop}{\hspace{.2cm}$|$\hspace{.2cm}}
    \newcommand{\population}{population\xspace}
    \newcommand{\system}{system\xspace}

    \newcommand{\broadcast}{Bit Dissemination\xspace}
    
    \newcommand{\outputbit}{opinion\xspace}
    \newcommand{\outputbits}{opinions\xspace}
    \newcommand{\inputbit}{input bit\xspace}
    \newcommand{\inputbits}{input bits\xspace}
    \newcommand{\majority}{Majority\xspace}
    \newcommand{\prota}{\textsc{Syn-Generic}\xspace}
    \newcommand{\protb}{\textsc{P}\xspace}
    \newcommand{\protc}{\textsc{Syn-P}\xspace}
    \newcommand{\clocksyncronization}{Clock Synchronization\xspace}

    \newcommand{\majprot}{\textsc{Syn-Phase-Spread}\xspace}
    \newcommand{\majprotgiven}{\textsc{Phase-Spread}\xspace}
    \newcommand{\synclock}{\textsc{Syn-Clock}\xspace}

    \newcommand{\major}{\textsc{maj-consensus}\xspace}
    \newcommand{\bigO}{\mathcal O}

    \newcommand{\mode}{\textsc{maj}}

    \newcommand{\simple}{\textsc{Syn-Simple}\xspace}
    \newcommand{\intermediate}{\textsc{Syn-Intermediate}\xspace}
    \newcommand{\three}{\textsc{Syn-$3$Bits}\xspace}

    \newcommand\NN{\mathbb{N}}
    \newcommand\RR{\mathbb{R}}

    \renewcommand{\varepsilon}{\epsilon}
    \newcommand{\counter}{clock\xspace}
    \newcommand{\counters}{clocks\xspace}

    \global\long\def\typeblack{\mathcal{B}}
    \global\long\def\typewhite{\mathcal{W}}
    \global\long\def\sizeblack{k_1}
    \global\long\def\sizewhite{k_0}
    \global\long\def\bitmaj{b_{maj}}
    \global\long\def\bitfirst{b_{1}}
    \newcommand{\constpower}{\gamma}
    \newcommand{\const}{\gamma}

    \newcommand{\agreementQ}{\const \log n + \const \log \log T}

    \global\long\def\nowblack{\sizeblack^{\left(i\right)}}
    \global\long\def\nowwhite{\sizewhite^{\left(i\right)}}
    \global\long\def\boostblack{\sizeblack^{\left(1\right)}}
    \global\long\def\boostwhite{\sizewhite^{\left(1\right)}}
    \global\long\def\nextblack{\sizeblack^{\left(i+1\right)}}
    \global\long\def\nextwhite{\sizewhite^{\left(i+1\right)}}

    \global\long\def\protocol{\textsc{P}\xspace}
    \newcommand{\emul}{\textsc{Emul(\protocol)}\xspace}
    \newcommand{\independence}{bitwise-independence\xspace}
    \newcommand{\bitwisemodel}{\mathcal{BIT}}
    \global\long\def\runningtime{L_{\protocol} }
    
    \global\long\def\runningtimethree{L_{\synclock}}
    \newcommand{\emulround}{z}

\title{Minimizing Message Size in Stochastic Communication Patterns:\\
    Fast Self-Stabilizing Protocols with 3~bits\footnote{A preliminary version of this work appears as a 3-pages Brief Announcement in PODC 2016 \cite{BKN16} and as an extended abstract at SODA 2017 \cite{BKN17}.}}

\author{Lucas Boczkowski
    \thanks{
    IRIF, CNRS and University Paris Diderot, Paris, 75013, France. E-mail: {\tt \{Lucas.Boczkowski,Amos.Korman\}@irif.fr}.} 
    \and
    Amos Korman
    \footnotemark[2]
    \and
    Emanuele Natale
    \thanks{
    Max Planck Institute for Informatics, Saarbrücken, 66123 , Germany. E-mail: {\tt  emanuele.natale@mpi-inf.mpg.de}.
    This work has been partly done while the author was visiting the
    Simons Institute for the Theory of Computing. \hfill \break 
    \smallskip
    This work has received funding from the European Research Council (ERC)
    under the European Union's Horizon 2020 research and innovation programme
    (grant agreement No 648032).}
}
\date{}

\begin{document}

\maketitle

\begin{abstract}
    This paper considers the basic $\pull$ model of communication, in which in
    each round, each agent extracts  information from few randomly chosen
    agents. 
    We seek to identify the smallest amount of information revealed in each
    interaction (message size) that nevertheless allows for efficient and
    robust computations of fundamental information dissemination tasks. 
    We focus on the {\em \majority \broadcast} problem that considers a
    population of $n$ agents, with a designated subset of {\em source agents}.
    Each source agent holds an {\em input bit} and each agent holds an {\em
    output bit}. The goal is to let all agents converge their output bits on
    the most frequent input bit of the sources (the {\em majority bit}). Note
    that the particular case of a single source agent corresponds to the
    classical problem of {\em Broadcast} (also termed {\em Rumor Spreading}).
    We concentrate on the severe fault-tolerant context of {\em
    self-stabilization}, in which a correct configuration must be reached
    eventually, despite all agents starting the execution with arbitrary
    initial states. In particular, the specification of who is a source and
    what is its initial input  bit may be set by an adversary.

    We first design a general compiler which can essentially transform any
    self-stabilizing algorithm with a certain property (called ``the {\em
    \independence property}'') that uses $\ell$-bits messages to one that uses
    only $\log \ell$-bits messages, while paying only a small penalty in the
    running time. By applying this compiler recursively we then obtain a
    self-stabilizing {\em Clock Synchronization} protocol, in which  agents
    synchronize their clocks modulo some given integer $T$, within
    $\tilde{\bigO}(\log n\log T)$ rounds w.h.p., and using messages that
    contain $3$ bits only.
     
    We then employ the new Clock Synchronization tool to obtain a
    self-stabilizing \majority \broadcast protocol which converges in
    $\tilde{\bigO}(\log n)$ time, w.h.p., on every initial configuration,
    provided that the ratio of sources supporting the minority opinion is
    bounded away from half. Moreover, this protocol also uses only 3 bits per
    interaction. 
\end{abstract}

\section{Introduction}

\subsection{Background and motivation}\label{sub-background}

Distributed systems composed of  limited agents that  interact in a stochastic
fashion to jointly perform tasks are common in the natural world as well as in
engineered  systems. Examples  include a wide range of insect populations
\cite{HM85}, chemical reaction networks \cite{CCDS14}, and
mobile~sensor~networks \cite{Pop1}. Such systems have been studied in various
disciplines, including biology, physics, computer science and chemistry, while
employing different mathematical and experimental tools. 
%For example, using
%computer simulations to model animal group interactions, Couzin et al.
%demonstrated how groups can reach majority-consensus decisions, even though
%informed individuals do not know whether they are in a majority or minority
%\cite{COU05}. 

From an algorithmic perspective, such complex systems share a number of
computational challenges. Indeed, they  all perform collectively in
dynamically changing environments despite being composed of limited individuals
that communicate through seemingly unpredictable, unreliable, and  restricted
interactions. Recently, there has been significant interest in understanding
the  computational limitations that are inherent to such systems, by
abstracting some of their characteristics as distributed computing models, and
analyzing them algorithmically \cite{Dan2,Pop1,AFJ06,BCN15,Doty,OHK14}.
These models usually consider  agents which are restricted in their memory and
communication capacities, that interact independently and uniformly at random
(u.a.r.). By now, our understanding of the computational power of such models
is rather advanced. However, it is important to note that much of this
progress has been made assuming non-faulty scenarios - a rather strong
assumption when it comes to natural or sensor-based systems. 
For example, to synchronize actions between processors, many known distributed
protocols rely on the assumption that processors know when the protocol is
initiated. However, in systems composed of limited individuals that do not
share a common time notion, and must react to a  dynamically changing
environment, it is often unclear how to achieve such conditions. To have a
better  understanding of such systems, it is  desirable to identify the weakest
computational models that still allow for both efficient as well as robust
computations. 

This paper concentrates on the basic $\pull$ model of communication
\cite{DGH88,DF11a,DGM11,KSSV00}, in which in each round, each agent can
extract (pull) information from few other agents, chosen u.a.r. In the computer
science discipline, this model, as well as its companion $\push$ model, gained
their popularity due to their simplicity and  inherent robustness to different
kinds of faults. Here, focusing more on the context of natural systems, we view
the $\pull$ model as an abstraction for communication in well-mixed scenarios,
where agents can occasionally ``observe'' arbitrary other agents.
%\footnote{This
%may relate to the notion of {\em passive communication} commonly used by
%biologists to refer to communication that is based on observing the
%behavior of other individuals \cite{WIL92} (in contrast to {\em active
%communication} in which agents ``deliberately'' signal other agents).}. 
%
We aim to identify the minimal model requirements with respect to
achieving basic information dissemination tasks under conditions of increased
uncertainty. 
As many natural systems appear to be more restricted by their communication
abilities than by their memory capacities \cite{Survey,beeping1,emek}, our main
focus is on understanding what can be computed  while revealing as few bits per
interaction as possible\footnote{We note that stochastic communication patterns
such as $\pull$ or $\push$ are  inherently sensitive to congestion issues.
Indeed, in such models it is unclear how to simulate a protocol that uses large
messages while using only small size messages. For example, the straightforward
strategy of breaking a large message into small pieces and sequentially sending
them one after another  does not work, since one typically cannot make sure
that the small messages reach the same destination. Hence, reducing the message
size  may have a profound impact on the running time, and perhaps even on the
solvability of the problem at hand.}.
\medskip

\paragraph{Self-stabilizing  \broadcast.}

Disseminating information from one or several sources to the rest of the
\population is one of the most fundamental building blocks in distributed
computing \cite{CHHKM12,CLP09,DGH88,DF11a,KSSV00}, and an important primitive
in natural systems \cite{FishConsensus,Razin,ManyEyes}. 
Here, we focus on the {\em \majority \broadcast} problem defined as follows. We consider a \population of $n$ agents.
The \population may contain multiple {\em source agents} which are specified by
a designated bit in the state of an agent indicating whether the agent is a
source or not. 
Each source agent holds a binary {\em \inputbit}, however, sources may not
necessarily agree on their \inputbits. In addition, each agent holds a binary
{\em output bit} (also called {\em \outputbit}).
% The \inputbits of source agents may be considered as more reliable or up-to-date, although these bits may be conflicting. 
The goal of all agents is to converge their \outputbit on the majority bit
among the initial \inputbits of the sources, termed $\bitmaj$. This problem
aims to capture scenarios in which some individuals view themselves as
informed, but some of these agents could also be wrong, or not up-to-date. Such
situations are common in nature  \cite{COU05,Razin} as well as in man-made
systems. The number of sources is termed $k$. We do not assume that agents know
the value $k$, or that sources know whether they are in the majority or
minority (in terms of their \inputbit). For simplicity, to avoid dealing with
the case that the fraction of the majority \inputbit among sources is
arbitrarily close to that of  the minority \inputbit, we shall guarantee
convergence only when the fraction of source agents holding the majority
\inputbit is bounded away from $1/2$. 

The particular case where we are promised to have $k=1$ is called \broadcast,
for short. In this case we have a single source agent that aims to disseminate
its \inputbit $b$ to the rest of the \population, and there are no other
sources introducing a conflicting opinion. Note that this  problem has been
studied extensively in different models under different names (e.g., {\em
Broadcast} or {\em Rumor Spreading}). 
A classical example of \broadcast considers the  synchronous $\push$/$\pull$
communication model, where $b$ can be propagated from the source  to all other
agents in $\bigO(\log n)$ rounds, by simply letting each uninformed agent copy
it whenever it sees an informed agent \cite{KSSV00}. The correctness of this
protocol heavily relies on the absence of incorrect information held by  agents. Such reliability however may be difficult to achieve in dynamic
or unreliable conditions. 
For example, if the source is sensitive to an unstable environment, it may
change its mind several times before stabilizing to its final \outputbit.
Meanwhile, it may have already invoked several consecutive executions of the
protocol with contradicting initial \outputbits, which may in turn ``infect''
other agents with the wrong \outputbit $1-b$. If agents do not share a common
time notion, it is unclear how to let infected agents distinguish their current
wrong \outputbit from the more ``fresh'', correct \outputbit.  
To address such difficulty, we consider the context of {\em
self-stabilization} \cite{dijkstra}, where agents must converge to a correct
configuration from any initial configuration of states.

\subsection{Technical difficulties and intuition}\label{sec:technical}

Consider the \broadcast problem (where we are guaranteed to have a single source agent). This particular case is already difficult in the self-stabilizing context if we are restricted to use $O(1)$ bits per interaction. As hinted above, a main difficulty lies  in the fact that agents do not necessarily share a common time notion. Indeed, it is easy to see that if all agents share the same clock, then convergence can be achieved in $\bigO(\log n)$ time with high probability (w.h.p.), i.e, with a probability of at least $1 - n^{-\Omega(1)}$, and using two bits per interaction. 
\medskip

\paragraph{Self-stabilizing \broadcast ($k=1$) with 2 bits per interaction, assuming synchronized clocks.}
The source sets her output bit to be her input bit $b$. In addition to communicate its output bit $b_u$, each agent $u$ stores and communicates a {\em certainty} bit~$c_u$. Informally, having a certainty bit equal to 1 indicates that the agent is certain of the correctness of its output bit. The source's certainty bit is always set to 1.  Whenever a non-source agent $v$ observes $u$ and sees the tuple $(b_u,c_u)$, where $c_u=1$,  it copies the output and certainty bits of $u$ ({\em i.e.,} sets $b_v=b_u$ and $c_v=1$). In addition,  all non-source agents count rounds, and reset their certainty bit to~0 simultaneously every $T=\bigO(\log n)$ rounds. The reset allows to get rid of ``old'' output bits
that may result from applying the protocol before the source's output bit has stabilized. This way, from the first time a reset is applied after the source's output bit has stabilized, the correct source's output bit will propagate to all agents within $T$ rounds, w.h.p. 
Note however, that if agents do not share a consistent notion of time they cannot  reset their certainty bit to zero simultaneously. In such cases, it is unclear how to prevent agents that have just reset their certainty bit to 0 from being ``infected'' by ``misleading'' agents, namely, those that have the wrong output bit and certainty bit equal to 1.
\medskip
\paragraph{Self-stabilizing \broadcast ($k=1$) with a single bit per interaction, assuming synchronized clocks.}
Under the assumption that all agents share the same clock, the following trick shows how to obtain convergence  in $\bigO(\log n)$ time and using only a single bit per message, namely, the output bit. As before, the source sets her output bit to be her input bit $b$. Essentially, agents divide time into phases of some prescribed length $T=\bigO(\log n)$, each of them being further subdivided into $2$ subphases of length $T/2$. In the first subphase of each phase, non-source agents are {\em sensitive} to \outputbit $0$. This means that whenever they see a $0$ in the output bit of another agent, they turn their output bit to $0$, but if they see 1 they ignore it. Then, in the second subphase of each phase, they do the opposite, namely they switch their output bit to $1$ as soon as they see a $1$ (see Figure \ref{fig:wheel}). %
Consider the first phase starting after initialization.  If $b=0$ then within one complete subphase $[1,T/2]$, every output bit is $0$ w.h.p., and remains there forever. Otherwise, if $b=1$, when all agents go over a subphase $[T/2+1, T]$ all output bits are set to $1$ w.h.p., and remain $1$ forever. Note that a common time notion is required to achieve correctness.

\begin{figure}[!ht]
    \centering
    \includegraphics[width=0.34\textwidth]{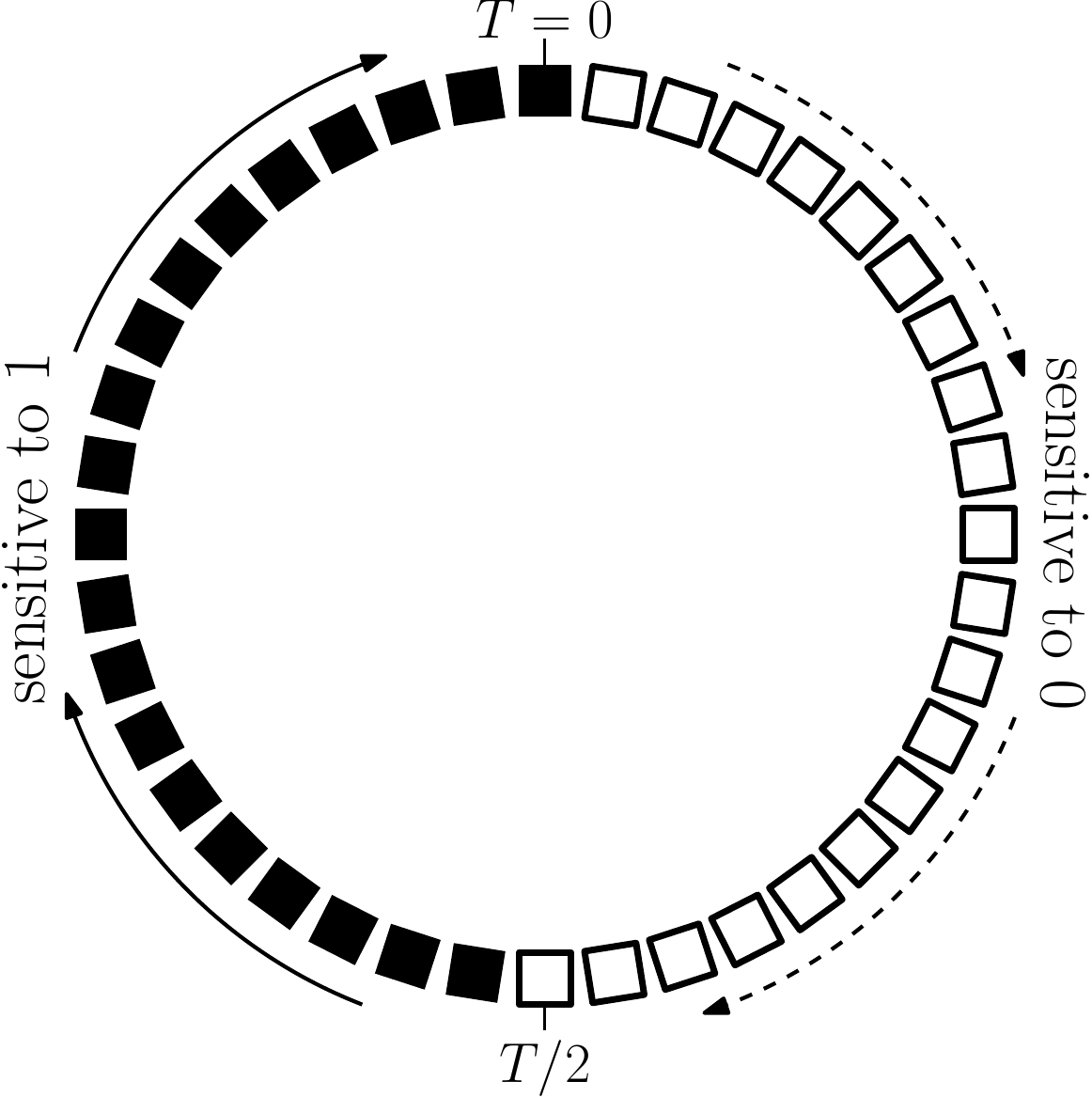}
    \caption{The division in subphases used for self-stabilizing \broadcast with a clock. During the first half, between times $1$ and $T/2$, agents are sensitive to $0$. Then they are sensitive to $1$.} \label{fig:wheel}
\end{figure}
 
The previous protocol indicates that the self-stabilizing \broadcast problem is highly related to the self-stabilizing {\em \clocksyncronization} problem, where each agent internally stores a \counter modulo $T=\bigO(\log n)$  incremented at every round and,  despite having arbitrary initial states, all agents should converge on sharing the same value of the \counter. Indeed, given such a protocol,  one can obtain a self-stabilizing \broadcast protocol by running the Clock Synchronization protocol in parallel to the last example protocol. This parallel execution costs only an additional  bit to the message size and a $\bigO(\log n)$ additive factor to the time complexity over the   complexities of the Clock Synchronization protocol. 
\medskip
%To synchronize clocks modulo $T$ in a self-stabilizing manner, one could use the stabilizing consensus protocol in \cite{DGM11}, by displaying all the bits of the clocks in each message, and reaching consensus on each of them separately and in parallel, while incrementing the clocks~(see Section~\ref{sec:simple} for more details). Unfortunately, this approach is wasteful in terms of message size, as it requires to reveal $\log T=\bigO(\log\log n)$ bits per interaction. As another approach, one could aim at sequentially synchronizing clocks~bit~after~bit. That is, first display and synchronize the first bit; then, once agents ``know'' that the first bit has been synchronized,  display and synchronize the second bit, etc. This approach is problematic in the context of self-stabilization, since, first, it requires agents to ``know'' when a bit is synchronized, and second, it requires agents to agree on the bit index that they currently aim to synchronize. Both of these seem to require clocks to be synchronized to begin with.

\paragraph{Intuition behind the self-stabilizing Clock Synchronization algorithm.}

Our technique for obtaining the Clock Synchronization protocol is based on a
compact recursive use of the stabilizing consensus protocol  proposed by Doerr
et al. \cite{DGM11} through our \corelem{} (Theorem \ref{lem:saving}). 

In the Preliminary section (Section \ref{sec:simple}) we describe a simple
protocol called \simple that uses $\bigO(\log T)$ bits per message. In \simple,
each agent $u$ maintains  a clock $C_u \in [0, T-1]$. At each round, each agent
$u$ displays the \outputbit of her clock, pulls~$2$ other such clock
\outputbits, and updates her clock as the bitwise majority of the two clocks
she pulled and her own. Then the clock $C_u$ is incremented.  This protocol
essentially amounts to running the protocol of Doerr et al. on each bit
separately and in parallel, and self-stabilizes in $\bigO(\log T \log n)$
rounds w.h.p. (Proposition \ref{prop:simple}).

We want to apply a strategy similar to \simple, while using only $\bigO(1)$
many bits per interaction.
The core technical ingredient, made rigorous in \corelemma, is that a certain
class of protocols using messages of $\numbits$ bits, to which \simple belongs,
can be emulated by another protocol which uses $\lceil \log \numbits \rceil +1$
bits only.
The idea is to build a \counter modulo $\numbits$ using \simple itself on
$\lceil \log \numbits \rceil$ bits and sequentially display one bit of the
original $\numbits$-bit message according to such clock. Thus, by applying such
strategy to \simple itself, we use a smaller \counter modulo $\numbits' \ll
\numbits$ to synchronize a \counter modulo $\numbits$. Iterating such process,
in Section \ref{sub:powerof}, we obtain a compact protocol which uses only $3$
bits. 
%
%At this point,
%$T$ needs to be a power of $2$ the running time is not optimal in terms of $T$ and we improve on this in Section \ref{sub:powerof}. 

\subsection{The model}
\label{sec:model}
\bigskip
\paragraph{The communication model.}  

We adopt the the synchronous $\pull$ model \cite{BCN16,DGH88}. Specifically, in the $\pull(\mesnum)$ model, communication proceeds in discrete rounds.
In each round, each agent $u$ ``observes'' $\mesnum$ arbitrary other agents, chosen u.a.r.%
\footnote{``u.a.r.'' stands for ``uniformly at random'', with replacement.}
among all agents, including herself.
(We often omit the parameter $\eta$ when it is equal to 2). 
When an agent $u$ ``observes'' another agent $v$, she can peek into a designated {\em visible part} of $v$'s memory. If several agents observe an agent $v$ at the same round then they all see the same visible part. The {\em message size} denotes the  number of bits stored in the visible part of an agent. We denote with $\pull(\mesnum,\numbits)$
the $\pull(\mesnum)$ model with message size $\numbits$. 
We are primarily interested in message size that is independent of $n$, the number of agents.
\medskip

\paragraph{Agents.} 
We assume that agents do not have unique identities, that is, the system is {\em anonymous}. We do not aim to minimize the (non-visible) memory requirement of the agent, yet, we note that our constructions can be implemented with relatively short memory, using $\bigO(\log\log n)$ bits. We assume that each agent internally stores a \counter modulo some integer $T=\bigO(\log n)$, which is incremented at every round.
\medskip

\paragraph{\majority \broadcast problem.}

We assume a \system of $n$ agents each having an internal state that contains an {\em indicator bit} which indicates whether or not the agent is a {\em source}. Each source holds a binary {\em \inputbit}\footnote{
Note that having the indicator bit equal to 1 is equivalent to possessing an input bit: both are exclusive properties of source nodes. However, we keep them distinct for a clearer presentation.
}and each agent (including sources) holds a binary {\em \outputbit}. 
The number of sources (i.e., agents whose indicator bit is $1$) is denoted by $k$. We denote by $\sizewhite$ and $\sizeblack$ the number of sources whose \inputbit is initially set to  $1$ and $0$, respectively.
Assuming $\sizeblack\neq\sizewhite$, we define the \emph{majority bit}, termed $\bitmaj$, as $1$ if $\sizeblack >\sizewhite$ and $0$ if $\sizeblack<\sizewhite$.
Source agents know that they are sources (using the indicator bit) but they do not know whether they hold the majority bit. The parameters $k$, $\sizeblack$ or $\sizewhite$ are not known to the sources or to any other agent.  
It is required that the \outputbits of all agents eventually converge to the majority bit%
%\footnote{The majority is not defined if $\sizeblack=\sizewhite$; in this case the only requirement is consensus, i.e. that the outputs of the agents are eventually equal.}
$\bitmaj$. 
We note that agents hold their output and indicator bits privately, and we do not require them to necessarily reveal these bits publicly (in their visible parts) unless they wish to. 
To avoid dealing with the cases where the number of sources holding the majority bit is arbitrarily close to $\frac k2$, we shall guarantee correctness (w.h.p.)
only if the fraction of sources holding the majority is bounded away from~$\frac{1}{2}$, i.e., only if $|\frac{\sizeblack}{\sizewhite}-1|>\epsilon$, for some positive constant $\epsilon$.  
When $k=1$, the problem is called {\em \broadcast}, for short. Note that in this  case, the single source agent holds the bit $\bitmaj$ to be disseminated and there is no other source agent introducing a conflicting opinion.
\medskip
 
\paragraph{$T$-\clocksyncronization.}

Let $T$ be an integer. In the {\em $T$-\clocksyncronization} problem, each agent maintains a {\em \counter} modulo $T$ that is incremented at each round. The goal of agents is to converge on having the same value in their \counters modulo $T$. (We may omit the parameter $T$ when it is clear from the context.)
\medskip

\paragraph{Probabilistic self-stabilization and convergence.}

Self-stabilizing protocols are meant to guarantee that  the \system eventually 
converges to a \emph{legal} configuration regardless of the initial states of the
agents \cite{dijkstra}. Here we use a slightly weaker notion, called {\em probabilistic self-stabilization}, where stability is guaranteed w.h.p. \cite{BCN15b}. More formally, for the  \clocksyncronization problem, we assume that {\em all} states are initially set by an adversary. For the  \majority
\broadcast problem, we assume that {\em all} states are initially set by an adversary except that it is assumed 
that the agents know their total number $n$, and that this information is not corrupted.

In the context of $T$-\clocksyncronization, a legal configuration is reached when all clocks show the same time modulo~$T$, and in the \majority \broadcast problem, a legal configuration is reached when all agents output  the  majority bit $\bitmaj$.
Note that in the context of the  \majority 
\broadcast problem, the legality criteria depends on the initial configuration (that may be set by an adversary). That is, the agents must converge their \outputbit on the majority of input bits of sources, as evident in the initial configuration.

The \system is said to {\em stabilize} in $t$ rounds if, from any initial configuration w.h.p.,
within $t$ rounds it reaches a legal configuration and remains legal for at least some polynomial time  \cite{BCN15b,BCN16,DGM11}. In fact, for the self-stabilizing \broadcast problem, if there are no conflicting source agents holding a minority opinion (such as in the case of a single source agent), then our protocols guarantee that once a legal configuration is reached, it remains legal~indefinitely. 
Note that, for any of the problems, we do not require that each agent irrevocably commits to
a final \outputbit but that eventually agents arrive at a 
legal configuration without necessarily being aware of that.

\subsection{Our Results}\label{sec:our-results}

Our main results are the following.

\begin{theorem}\label{thm:multi}
    Fix an arbitrarily small constant $\epsilon>0$. There exists a  protocol, called \majprot, which solves the \majority \broadcast problem in a self-stabilizing manner in $\tilde{\bigO}(\log n)$ rounds%
    \footnote{With a slight abuse of notation, with $\tilde{\bigO}(f(n)g(T))$ we refer to $f(n)g(T) \log^{\bigO(1)}(f(n)) \log^{\bigO(1)}(g(T))$. All logarithms are in base $2$.}
    w.h.p using $3$-bit messages, provided that the majority bit is supported by at least a fraction $\frac 12 + \epsilon$ of the source agents.
\end{theorem}

\noindent Theorem \ref{thm:multi} is proved in Section \ref{sec:majority}.
The core ingredient of \majprot is our construction of an efficient self-stabilizing $T$-\clocksyncronization protocol, which is used as a black-box. For this purpose, the case that interests us  is when  $T=\tilde\bigO(\log n)$. 
Note that in this case, the following theorem, proved in  Section \ref{sec:clock}, states that the convergence time of the Clock Synchronization algorithm is $\tilde{\bigO}(\log n)$.

\begin{theorem}\label{thm:main}
    Let $T$ be an integer. There exists a self-stabilizing $T$-\clocksyncronization protocol, called \synclock, which
    employs only 3-bit messages, and synchronizes \counters modulo~$T$ within $\tilde{\bigO}(\log n \log T)$ rounds w.h.p. 
\end{theorem}

\noindent In addition to the self-stabilizing context our protocols can tolerate the presence of Byzantine agents, as long as their number is%
\footnote{Specifically, it is  possible to show that, as a corollary of our
    analysis and the fault-tolerance property of the analysis in \cite{DGM11},
    if $T \leq poly(n)$ then \synclock can tolerate the presence of up to
    $\bigO(n^{1/2 - \varepsilon})$ Byzantine agents for any $\epsilon>0$. In
    addition, \majprot can tolerate
    $\min\{(1-\epsilon)(k_{maj}-k_{min}),n^{1/2-\epsilon}\}$ Byzantine agents,
    where $k_{maj}$ and $k_{min}$ are the number of sources supporting the
    majority and minority opinions, respectively. Note that for the case of a
    single source ($k=1$), no Byzantine agents are allowed; indeed, a single
    Byzantine agent pretending to be the source with the opposite \outputbit can
    clearly ruin any protocol.}
$\bigO(n^{1/2 - \varepsilon})$. 
However, in order to focus on the self-stabilizing aspect of our results, in this work we do not explicitly address the presence of Byzantine agents.

The proofs of both Theorem \ref{thm:main} and Theorem \ref{thm:multi} rely on
recursively applying a new general compiler which can essentially transform
any self-stabilizing algorithm with a certain property (called ``the
\independence property'') that uses $\numbits$-bit messages to one that uses
only $\lceil \log \ell\rceil +1$-bit messages, while paying only a small
penalty in the running time. This compiler is described in Section
\ref{sec:compiler}, in Theorem \ref{lem:saving}, which is also referred as
``\corelemma''. The structure between our different lemmas and results is
summarized in the picture below, Figure \ref{fig:bigscheme}.

It remains an open problem, both for the self-stabilizing \broadcast problem and for the self-stabilizing Clock Synchronization problem,  whether the message size can be reduced to 2 bits or even to 1 bit, while keeping the running time poly-logarithmic. 

%\paragraph{Roadmap to our results.}
\begin{figure*}[!ht]
    \centering
    \includegraphics[width=0.9\textwidth]{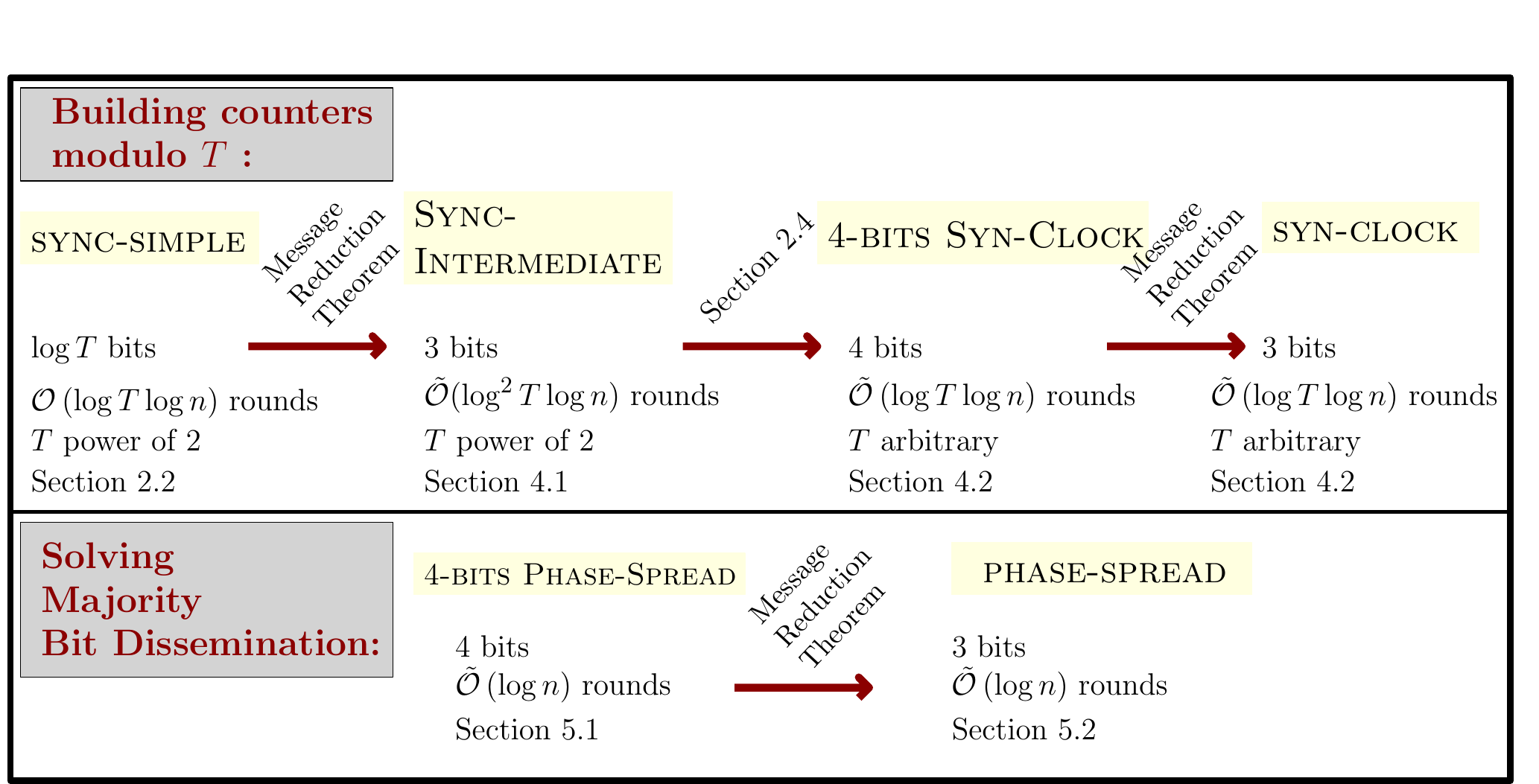}
    \caption{The structure of our arguments. Note that the Message Reduction Theorem  is used on three occasions.} \label{fig:bigscheme}
\end{figure*}

\subsection{Related work}

The computational study of abstract systems composed of simple individuals that interact using highly restricted and stochastic interactions has recently been gaining considerable attention in the community of theoretical computer science. Popular models include {\em population protocols} \cite{Pop1,Pop2,Pop-ss,Joffroy}, which typically consider constant  size individuals that interact in pairs (using constant size messages) in random communication patterns, and the {\em beeping} model \cite{beeping1,emek}, which assumes  a fixed network with extremely restricted communication. Our model also falls in this framework as we consider the $\pull$ model \cite{DGH88, KSSV00, KDG03} with constant size messages. So far, despite interesting works that consider different fault-tolerant contexts \cite{Aspnes,Pop-ss,Joffroy}, most of the progress in this framework considered non-faulty scenarios.

Information dissemination is one of the most well-studied topics in the community of distributed computing, see, {\em e.g.,} \cite{Aspnes,CHHKM12,DGH88,DF11a,DGM11,OHK14,KSSV00}. Classical examples include the {\em Broadcast} (also referred to in the literature as  {\em Rumor Spreading}) problem, in which a piece of information residing at one source agent is to be disseminated to the rest of the \population, and  {\em majority-consensus} (here, called \majority \broadcast) problems in which processors are required to agree on a common output value which is the majority initial input value among all agents~\cite{Aspnes,kutten} or among a set of designated source agents \cite{OHK14}. An extensive amount of research has been dedicated to study such problems in $\push / \pull$ based protocols (including the {\em phone call} model), due to the inherent simplicity and fault-tolerant resilience of such meeting patterns. Indeed, the robustness of $\push / \pull$ based protocols to weak types of faults, such as crashes of messages and/or agents, or to the presence of relatively few Byzantine agents, has been known for quite a while \cite{ES09,KSSV00}. 
Recently, it has been shown that under the $\push$ model, there exist  efficient Broadcast and \majority \broadcast protocols that use a single bit per message and can overcome flips in messages (noise) \cite{OHK14}. The protocols therein, however, heavily rely on the assumption that agents know when the protocol has started. Observe that in a self-stabilizing context, in which the adversary can corrupt the initial \counters setting them to arbitrary times, such an assumption would be difficult to remove while preserving the small message size.

In general, there are only few known self-stabilizing protocols that operate
efficiently under stochastic and capacity restricted interactions. An example,
which is also of high relevance to this paper, is the work of Doerr et al. on
{\em Stabilizing Consensus} \cite{DGM11} operating in the  $\pull$ model. In
that work, each agent initially has a state taken out of a set of $m$
\outputbits and the goal is to converge on one of the proposed states. The
proposed algorithm which runs in logarithmic time is based on sampling the
states of $2$ agents and updating the agent's state to be the median of the $2$
sampled states and the current state of the agent ($3$ \outputbits in total).
Since the total number of possible states is $m$, the number of bits that must
be revealed in each interaction is $\Omega(\log m)$. 
Another example is the plurality consensus protocol in \cite{BCN16}, in which each agent has initially an opinion and we want the system to converge to the most frequent one in the initial configuration of the system.
In fact, the \majority \broadcast problem can be viewed as a generalization of the {\em majority-consensus} problem (i.e. the plurality consensus problem with two opinions), to the case in which multiple agents may initially be unopinionated. In the previous sense, we also contribute to the line of research on the majority-consensus problem  \cite{BCN15, CER15, EFK16}.

Another fundamental building block is Clock Synchronization  \cite{Attiya,Lamport,Lenzen3,Lenzen2}. 
We consider a synchronous system in which clocks tick at the same pace but may
not share the same \outputbit. This version  has earlier been studied in e.g.,
\cite{Ben-Or,Dolev07,Dolev97,Dolev04,SIROCCO,Herman} under different names,
including ``digital Clock Synchronization'' and ``synchronization of
phase-clocks''; We simply use the term ``Clock Synchronization''. There is by
now a substantial line of work on Clock Synchronization problems in a
self-stabilizing context \cite{DolevKLRS13, Dolev04, Lenzen4, Lenzen5}. We note
that in these papers the main focus is on the resilience to Byzantine agents.
The number of rounds and message lengths are also minimized, but typically as a
function of the number of Byzantine processors. Our focus is instead on
minimizing the time and message complexities as much as possible. The authors
in \cite{Lenzen4, Lenzen5} consider mostly a deterministic setting. The
communication model is very different than ours, as every agent gets one
message from every other agent on each round. Moreover, agents are assumed to
have unique identifiers. In contrast, we work in a more restricted, yet
randomized communication setting. In \cite{DolevKLRS13, Lenzen4} randomized
protocols are also investigated. We remark that the first protocol we discuss
\simple (Proposition \ref{prop:simple}), which relies on a known simple
connection between consensus and counting \cite{DolevKLRS13}, already improves
exponentially on the randomized algorithms from \cite{DolevKLRS13, Lenzen4} in
terms of number of rounds, number of memory states, message length and total
amount of communication, in the restricted regime where the resilience
parameter $f$ satisfies $ \log n \leq f \leq \sqrt{n}$. We further note that
the works \cite{Lenzen5,Lenzen4} also use a recursive construction for their
\counters (although very different from the one we use in the proof of
Theorem~\ref{thm:main}). The induction in \cite{Lenzen4} is on the resilience
parameter $f$, the number of agents and the \counter length together.
This idea is improved in \cite{Lenzen5}  to achieve optimality in terms of
resilience to Byzantine agents.

To the best of our knowledge there are no previous works on self-stabilizing Clock Synchronization or Majority Bit Dissemination that aim to minimize the message size beyond logarithmic.

\section{Preliminaries}\label{sec:Preliminaries}

\subsection{A majority based, self-stabilizing protocol for consensus on one bit}

Let us recall%
\footnote{Our protocols will use this protocol as a \emph{black box}. However,
we note that  the constructions we outline are in fact independent of the
choice of consensus protocol, and this protocol could be replaced by other
protocols that achieve similar guarantees.} 
the stabilizing consensus protocol by Doerr et al. in \cite{DGM11}.  
In this protocol, called \major, each agent holds  an \outputbit. In  each
round each agent looks at the \outputbits of two other random agents and
updates her \outputbit  taking the majority among the bits of the observed
agents and its own. 
Note that this protocol uses only a single bit per interaction, namely, the
\outputbit. The usefulness of \major comes from its extremely fast and
fault-tolerant convergence toward an agreement among agents, as given by the
following~result.

\begin{theorem}[Doerr et al. \cite{DGM11}]\label{thm:doerr}
    From any initial configuration,  \major converges to a state in which all
    agents agree on the same output bit in $\bigO(\log n)$ rounds, w.h.p.
    Moreover, if there are at most $\kappa \leq n^{1/2-\varepsilon}$ Byzantine
    agents, for any constant $\epsilon>0$, then after $\bigO(\log n)$ rounds
    all non-Byzantine agents have converged and consensus is maintained for
    $n^{\Omega(1)}$ rounds w.h.p.%
    \footnote{The original statement of \cite{DGM11} says that if at most
        $\kappa \leq \sqrt n$ agents can be corrupted at any round, then convergence
        happens for all but at most $\bigO(\kappa)$ agents. Let us explain how this
        implies the statement we gave, namely that we can replace $\bigO(\kappa)$
        by $\kappa$, if $\kappa \leq n^{\frac{1}{2}-\varepsilon}$. Assume that we
        are in the regime $\kappa \leq n^{\frac{1}{2}-\varepsilon}$. It follows
        from \cite{DGM11} that all but a set of $\bigO(\kappa)$  agents reach
        consensus after $\bigO(\log n)$ round. This set of size $\bigO(\kappa)$
        contains both Byzantine and non Byzantine agents. However, if the number of
        agents holding the minority \outputbit is $\bigO(\kappa ) =
        \bigO(n^{1/2-\varepsilon})$, then the expected number of non Byzantine
        agents that disagree with the majority \emph{at the next round} is in
        expectation $\bigO(\kappa^2 /n) = \bigO(n^{- 2\varepsilon} )$.
        Thus, by Markov's inequality,  this  implies, that at the next round
        consensus is reached among \emph{all non-Byzantine agents} w.h.p. Note also
        that, for the same reasons, the Byzantine agents do not affect any other
        non-Byzantine agent for $n^{\varepsilon} $ rounds w.h.p.}
    \label{thm:maj_consensus}
\end{theorem}

\subsection{Protocol \simple: A simple protocol with many bits per interaction}\label{sec:simple}
We now present a simple self-stabilizing $T$-Clock Synchronization protocol, called  \simple, that uses relatively many bits per message, and relies on the assumption that $T$ is a power of 2. The protocol is based on iteratively applying a self-stabilizing consensus protocol on each bit of the clock separately, and in parallel. 

Formally, each agent $u$ maintains  a clock $C_u \in [0, T-1]$.
At each round,  $u$ displays the \outputbit of her clock $C_u$, pulls~$2$ uniform other such clock \outputbits, and updates her clock as the bitwise majority of the 
two clocks it pulled, and her own.
Subsequently, the clock $C_u$ is incremented.
We present the pseudo code of \simple in Algorithm \ref{alg:simple}. 

\begin{pseudocode}[!ht]
    \textbf{\simple protocol}

    \smallskip
    \begin{description}
        \item[1] $u$ samples two agents $u_1$ and $u_2$.
        \item[2] $u$ updates its clock with the bitwise majority of its clock and those
            of the sample nodes.
        \item[3] $u$ increments its clock by one unit. \label{line:Qincr}
    \end{description}
    \caption{One round of \simple, executed by each agent $u$.}
    \label{alg:simple}
\end{pseudocode}

We prove the correctness of \simple in the next proposition.

\begin{proposition}\label{prop:simple}
    Let $T$ be a power of $2$. The protocol \simple is  a self-stabilizing protocol that uses $\bigO(\log T) $ bits per interaction and synchronizes clocks modulo $T$ in $\bigO(\log T \log n)$ rounds w.h.p. 
\end{proposition}
\begin{proof}
    Let us look at the least significant bit. One round of \simple is
    equivalent to one round of \major with an extra flipping of the \outputbit
    due to the increment of the clock. The crucial point is that all agents
    jointly flip their bit on every round. Because the function agents apply,
    \mode, is symmetric, it commutes with the flipping operation. More
    formally, let $\vec{b}_{t}$ be the vector of the first bits of the clocks
    of the agents at round $t$ under an execution of \simple.
    E.g. $(\vec{b}_{t})_u$ is the value of the less significant bit of node
    $u$'s clock at time $t$.
    Similarly, we denote by $\vec{c}_t$ the first bits of the clocks of the
    agents at round $t$ obtained by running a modified version of \simple in
    which {\em time is not incremented} (i.e. we skip line \ref{line:Qincr} in
    Algorithm \ref{alg:simple}). We couple $\vec{b}$ and $\vec{c}$ trivially,
    by running the two versions on the same interaction pattern (in other
    words, each agent starts with the same memory and pulls the same agents at
    each round in both executions).
    Then, $\vec{b_t}$ is equal to $\vec{c_t}$ when $t$ is even, while is equal to
    $\vec{b_t} = \mathbf{1}-\vec{c_t}$ when $t$ is odd.
    Moreover, we know from Theorem \ref{thm:doerr} that $\vec{c_t}$ converge to
    a stable \outputbit in a self-stabilizing manner. It follows that, from any
    initial configuration of states (i.e. clocks), w.h.p, after $\bigO(\log n)$
    rounds of executing \simple, all agents share the same \outputbit for their
    first bit, and jointly flip it in each round.
    Once agents agree on the first bit, since $T$ is a power of $2$, the
    increment of time makes them flip the second bit \emph{jointly} once every
    $2$ rounds%
    \footnote{To get the feeling of the kind of dependence more significant
    bits have on the less significant ones when $T$ is not a power of $2$
    observe that, for example, if $T=3$ then the first bit takes cyclically the
    values $1$, $0$ and again $0$.}.
    More generally, assuming agents agree on the first $\ell$ bits of their
    clocks, they \emph{jointly} flip the $\ell +1$'st bit once every $2^{\ell}$
    rounds, on top of doing the $\major$ protocol on that bit.
    Hence, the same coupling argument shows that the flipping doesn't affect the
    convergence on bit $\ell +1$. Therefore, $\bigO(\log n)$ rounds after the
    first $\ell$ bits are synchronized, w.h.p. the $\ell+1$'st bit is synchronized as
    well. The result thus follows by induction.
\end{proof}

\subsection{The \independence property}\label{sec:property}

Our general transformer described in Section \ref{sec:compiler} is useful for reducing the message size of protocols with a certain property called \independence. 
%Fortunately, all the protocols described in this paper enjoy this property. 
Before defining the property we need to define a variant of the $\pull$ model, which we refer to as the $\bitwisemodel$ model. The reason we introduce such a variant is mainly technical, as it appears naturally in our proofs.

Recall that in the $\pull(\mesnum,\ell)$ model, at any given round, each agent $u$ is reading an $\ell$-bit  message $m_{v_j}$ for each of the  $\mesnum$ observed agents $v_j$ chosen u.a.r. (in our case $\mesnum=2$), and then, in turn, $u$  updates her state according to the instructions of a protocol \protocol.
Informally, in the $\bitwisemodel$ model, each agent $u$ also receives $\mesnum$ messages, however, in contrast to the $\pull$ model where each such message corresponds to one observed agent, in the $\bitwisemodel$ model, the $i$'th bit of each such message is received independently from a (typically new) agent, chosen u.a.r. from all agents. 
\begin{definition}[The $\bitwisemodel$ model]\label{def-bits}
  In the $\bitwisemodel$ model,
  at each round, each agent $u$
  picks  $\mesnum\ell$ agents u.a.r., namely, $v_1^{(1)},v_2^{(1)}, \ldots v_{\numbits}^{(1)}$,\ldots,$v_1^{(\mesnum)},v_2^{(\mesnum)}, \ldots v_{\numbits}^{(\mesnum)}$, and reads $\hat{ s_i}^{(j)} = s_i(v_i^{(j)})$, the $i$-th bit of the visible part of agent $v_i^{(j)}$, for every $i \leq \numbits$ and $j \leq \mesnum$.
  For each $j \leq \mesnum$, let $\hat{m_j}(u)$ be the $\ell$-bit string $\hat{m_j}(u):= (\hat{s_1}^{(j)},\hat{s_2}^{(j)},\dots, \hat{s_\ell}^{(j)})$. 
  By a slight abuse of language we call the strings $\{\hat{m_j}(u)\}_{j\leq \mesnum}$
  the \emph{messages} received by $u$ in the $\bitwisemodel$ model.
\end{definition}

\begin{definition}[The $\independence$ property]\label{def:ip}
  Consider a protocol \protocol designed to work in the $\pull$ model. We say that \protocol has the \independence property if its correctness and running time guarantees remain the same, under the $\bitwisemodel$ model (assuming that given the messages $\{\hat{m_j}(u)\}_{j\leq \mesnum}$ it receives at any round, each agent $u$ performs the same actions that it would have, had it received these messages in the $\pull$ model).
\end{definition} 

Let us first state a fact about protocols having the \independence property.
\begin{lemma}\label{lem:ipfact}
    Assume protocol \prota is a protocol synchronizing \counters modulo $T$ for some $T$ and protocol \protb is a protocol which works assuming agents share a clock modulo $T$.
    Denote by \protc the parallel execution of \prota and \protb, with \protb using the clock synchronized by \prota. 
    If \prota and \protb are self-stabilizing then so is \protc, and the
    convergence time of \protc is at most the sum of convergence times of
    \prota and \protb.
    %The size of messages in \protc is the sum of the size of messages of \prota and \protb.
%
    Finally, if \prota and \protb have the \independence property, and \protb is also self-stabilizing, \protc has the \independence property too.
\end{lemma}
\begin{proof}
    The self-stabilizing property of \protc and its convergence time easily follows from those of \prota and \protb. 
    
    As for the \independence property, assume we run \protc in the
    $\bitwisemodel$ model. The execution of \prota is carried independently of
    the execution of \protb.
    Since, by hypothesis, \prota has the independence property, eventually all
    agents have a synchronized \counter modulo $T$. Thus, once \counters are
    synchronized, we can disregard the part of the message corresponding to
    \prota, and view the execution of \protc as simply \protb. 
    Therefore, since \protb is self-stabilizing and has the independence
    property, \protc still works in the $\bitwisemodel$ model as in the
    original $\pull$ model.
\end{proof}

\noindent We next show that the protocol \simple has the aforementioned \independence property.

\begin{lemma}\label{lem:simple_ind}
    \simple has the \independence property. 
\end{lemma}
\begin{proof}
    Let $\numbits'$ be the size of the \counters. 
    Assume the first $i<\numbits'$ bits of the \counters have been synchronized. At
    this stage, the $(i+1)$-st bit of each agent $u$ is flipped every $2^{i}$
    rounds and updated as the majority of the $(i+1)$-st bit of $C(u)$ and the $2$
    pulled messages on each round. 
    Since the first $\numbits'$ bits are synchronized, the previous flipping is
    performed by all agents at the same round.
    The thesis follows from the observation that, in order for \simple to work,
    we do not need the bit at index $(i+1)$ to come from the same agent as
    those bits used to synchronize the other indices, as long as convergence on
    the first $i$ bits has been achieved. 
\end{proof}

\section{A General Compiler that Reduces Message Size}\label{sec:compiler}

In this section we present a general compiler that allows to implement a protocol \protocol using $\numbits$-bit messages while using messages of order $\log \numbits$ instead, as long as \protocol enjoys the  \independence property. The compiler is based on replacing a message by an index to a given bit of the message. This tool will repeatedly be used in the following sections to obtain our Clock Synchronization and \majority \broadcast algorithms that use 3-bit messages.

\begin{theorem}[\corelemma]\label{lem:saving}
    Any self-stabilizing protocol \protocol in the $\pull(\mesnum,\ell)$ model having the \independence property, and whose running time is $\runningtime$, can be emulated by a protocol \emul which runs in the $\pull(2,\lceil \log (\subphasescorelem\ell) \rceil + 1)$ model, has running time $\bigO\left(\log (
    \mesnum \ell) \log n + \subphasescorelem \ell \runningtime \right)$ and has itself the \independence property. 
\end{theorem}
\begin{remark}
The only reason for designing \emul to run in the $\pull(2,\lceil \log (\subphasescorelem\ell)  \rceil + 1)$ model in \corelemma is the consensus protocol we adopt, \major, which works in the $\pull(2)$ model. In fact, \emul can be adapted to run in the $\pull(1,\lceil \log (\mesnum\ell)  \rceil + 1)$ model by using a  consensus protocol which works in the $\pull(1)$ model. However, no self-stabilizing binary consensus protocol in the $\pull(1)$ model with the same performances as \major is currently known.
\end{remark}
\proofof{Theorem \ref{lem:saving}}
    Let $s(u) \in \{0,1 \}^{\numbits}$ be the message displayed by an agent $u$
    under \protocol at a given round. For simplicity's sake, in the following
    we assume that $\eta$ is even, the other case is handled similarly. In
    \emul, agent $u$ keeps the message~$s(u)$ privately, and instead displays a
    \counter $\clock(u)$ written on $\lceil \log (\subphasescorelem\numbits)
    \rceil$  bits, and one bit of the message~$s(u)$, which we refer to as the
    $\protocol$-bit. 
    Thus, the total number of bits displayed by the agent operating in \emul is
    $\lceil \log (\subphasescorelem \numbits) \rceil +1$.
    The purpose of the \counter $\clock(u)$ is to indicate to agent $u$ which
    bit of $s(u)$ to display. In particular, if the counter has value $0$, then
    the $0$-th bit (i.e the least significant bit) of $s(u)$ is shown as the
    \protocol-bit, and so on. In what follows, we refer to $s(u)$ as the {\em
    private message} of $u$, to emphasize the fact that this message is not
    visible in \emul. See Figure \ref{fig:savingfig} for an illustration.

    \begin{figure*}[!ht]
        \centering
        \includegraphics[width=1\textwidth]{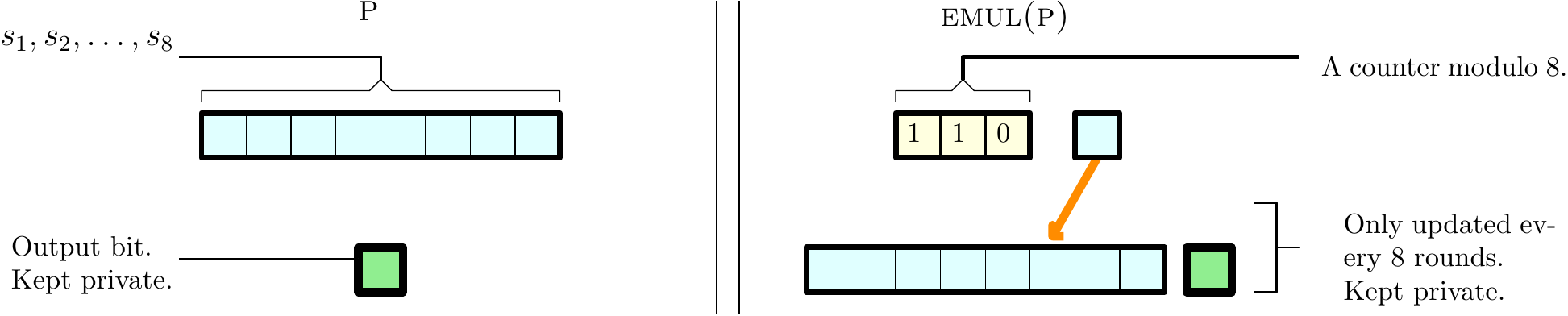}
        \caption{On the left is a protocol \protocol using $\numbits =8$ bits in
            total and pulling only one node per round ($\mesnum=1$). On the right
            is the emulated version \emul which uses $4$ bits only. The bits
            depicted on the bottom of each panel are kept privately, while the bits
            on the top are public, that is, appear in the visible part.} 
        \label{fig:savingfig}
    \end{figure*}

    Each round of \protocol executed in the $\pull(\mesnum,\ell)$ model by an
    agent $u$ is emulated by $ \subphasescorelem \numbits$ rounds of  \emul in
    the $\pull(2,\lceil \log (\subphasescorelem\ell) \rceil + 1)$ model. 
    We refer to such $ \subphasescorelem \numbits$ rounds as a {\em phase},
    which is further divided to $\subphasescorelem$ subphases of length
    $\numbits$. Note that since each agent samples 2 agents in a round, the
    total number of agents sampled by an agent during a phases is $\mesnum
    \numbits$.  
     
    For a generic agent $u$, a phase starts when its \counter $\clock(u)$ is
    zero, and ends after a full loop of its \counter (i.e. when $\clock(u)$
    returns to zero). Each agent $u$ is running protocol \simple on the $\lceil
    \log (\subphasescorelem \numbits) \rceil$  bits which correspond to her
    \counter $\clock(u)$. 
    Note that the phases executed by different agents may initially be
    unsynchronized, but, thanks to Proposition \ref{prop:simple}, the \counters 
    $\clock(u)$ eventually converge to the same value, for each agent $u$, and
    hence all agents eventually agree on when each phase (and subphase) starts. 

    Let $u$ be an arbitrary agent.
    Denote by $\hat s_1^{(1)}, \hat s_2^{(1)}, \ldots \hat s_\ell^{(1)}$, ...,
    $\hat s_1^{(\mesnum)}, \hat s_2^{(\mesnum)}, \ldots \hat
    s_\ell^{(\mesnum)}$  the \protocol-bits collected by $u$ from agents chosen
    u.a.r during a phase.
    Consider a phase and a round $\emulround\in \{1,\cdots, \subphasescorelem
    \numbits\}$ in that phase. Let $i$ and $j$ be such that $\emulround=j\cdot
    \numbits + i$. We view $\emulround$ as round $i$ of subphase $j+1$ of the
    phase. 
    On this round, agent $u$ pulls two messages from agents $v$ and $w$, chosen u.a.r.
    Once the clocks (and thus phases and subphases) have synchronized, agents
    $v$ and $w$ are guaranteed to be displaying the $i$th index of their
    private messages, namely, the values $s_i(v)$ and $s_i(w)$, respectively.
    Agent $u$ then sets $\hat s_i^{(2j-1)}$ equal to $ s_i(v)$ and
    %\footnote{If $\mesnum$ is odd, then $s_i(w)$ is simply discarded, since
    %one round of \protocol requires only $\mesnum$ messages.}
    $\hat s_i^{(2j)}$ equal to $s_i(w)$.

    In  \emul, the messages displayed by agents are only updated after a full
    loop of $C$. It therefore follows from the previous paragraph that the
    \protocol-bits collected by agent $u$ after a full-phase are distributed
    like the bits collected during one round of \protocol in the
    $\bitwisemodel$ model (see Definition \ref{def-bits}), assuming the
    \counters are synchronized already.

    \smallskip
    \noindent \emph{Correctness.} 
    The \independence property of \simple (Lemma \ref{lem:simple_ind}), implies
    that \simple still works when messages are constructed from the
    \protocol-bits collected by \emul. Therefore, from Proposition
    \ref{prop:simple}, eventually all the \counters $C$ are synchronized. 
    Since private messages $s$ are only updated after a full loop of $C$, 
    once the \counters $C$ are synchronized a phase of \emul corresponds to
    \emph{one} round of \protocol, executed in the $\bitwisemodel$ model.
    Hence, the hypothesis that \protocol operates correctly in a
    self-stabilizing way in the $\bitwisemodel$ model implies the correctness
    of \emul.

    \smallskip
    \noindent\emph{Running time.}
    Once the \counters $\clock(u)$ are synchronized, for all agents $u$, 
    using the first $\lceil \log (\subphasescorelem \numbits) \rceil$ bits of the messages,
    the agents reproduce an execution of \protocol with a multiplicative time-overhead of $\subphasescorelem \numbits$. 
    Moreover, from Proposition \ref{prop:simple}, synchronizing the clocks $C(u)$ takes $\bigO\left(\log (\mesnum m) \log n\right)$ rounds.
    Thus, the time to synchronize the \counters costs only an additive factor of $\bigO\left(\log (\mesnum m) \log n\right)$ rounds, and the total running time is  $\bigO\left(\log (\mesnum m) \log n \right) + \subphasescorelem \numbits \cdot \runningtime$.

    \smallskip
    \noindent\emph{Bitwise-independence property.}
    Protocol \emul inherits the \independence property from that of \simple
    (Lemma \ref{lem:simple_ind}) and \protocol (which has the property by
    hypothesis):
    We can apply Lemma \ref{lem:ipfact} where \prota is \simple and \protb is
    the subroutine described above, which displays at each round the bit of $P$
    whose index is given by a synchronized \counter $\clock$ modulo $\numbits$
    (i.e. the one produced by \simple). Observe that the aforementioned
    subroutine is self-stabilizing, since it emulates \protocol once \counters
    are synchronized. Then, in the notation of Lemma \ref{lem:ipfact}, \emul is
    \protc. 
\qed

\section{Self-Stabilizing Clock Synchronization}
\label{sec:clock}
In Section \ref{sec:simple} we described \simple ~-~ a simple self-stabilizing Clock Synchronization protocol that  uses $\log T$ bits per interaction. In this section we describe our main self-stabilizing Clock Synchronization protocol, \three, that uses only $3$ bits per interaction.
We first assume $T$ is a power of $2$. We show  how to get rid of this assumption in Section~\ref{sub:powerof}. 

\subsection{Clock Synchronization with $3$-bit messages, assuming $T$ is a power of two}
\label{sec:analysisthree}

\begin{figure*}[!ht]
    \centering
    \includegraphics[width=0.7\textwidth]{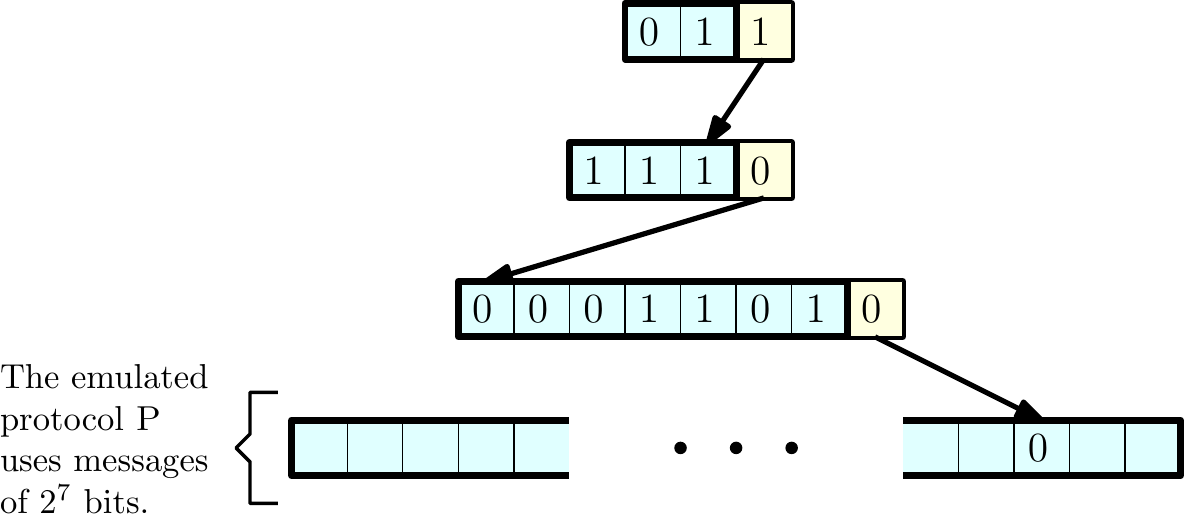}
    \caption{A more explicit view of our $3$-bit emulation of protocol
        \protocol, obtained by iterating Lemma \ref{lem:saving}. The down-most
        layer represents the $2^7$-bits message displayed by protocol
        \protocol. Each layer on the picture may be seen as the message of a
        protocol emulating \protocol with fewer bits, that is, as we go up on
        the figure we obtain more and more economical protocols in terms of
        message length. In particular, the top layer represents the  $3$-bit
        message in the final emulation. The left-most part of each message
        (colored in light blue) encodes a \counter. The right-most bit (colored
        in light yellow) of each message (except the bottom-most one)
        corresponds to a particular bit of the layer \emph{below} it. The index
        of this particular displayed bit is given by the value of the \counter.
        Each \counter on an intermediate layer is updated only when the
        \counter on the layer \emph{above} completes a loop (i.e., has value
        $0$). The \counter on the top-most layer is updated on every round.} 
    \label{fig:clocklayers}
\end{figure*}

In this section, we show the following result.
\begin{lemma}\label{lem:threebits}
    Let $T$ be a power of $2$. There exists a synchronization protocol
    \intermediate which synchronizes \counters modulo $T$ in time $ \tilde{\bigO}\left(\log^2 T \log n\right)$ using only 3-bit messages. Moreover, \intermediate has the \independence property.
\end{lemma}

Before presenting the proof of Lemma \ref{lem:threebits}, we need a remark about \counters.
\begin{remark}\label{rem:counters}
    In order to synchronize a clock $C$ modulo $T$, throughout the analysis we
    often obtain a clock $C'$ modulo $T$ \emph{which is incremented every $\ell$
    rounds}. However, $C'$ can still be translated back to a clock modulo
    $T$ which is incremented every round, by keeping a third clock $C''$ modulo
    $\ell$ and setting 
    \[
        C = C' + C'' \mod T.
    \]
\end{remark}

\proofof{Lemma \ref{lem:threebits}}
    At a high level, we simply apply iteratively \corelemma in order to reduce
    the message to $3$ bits, starting with \protocol = \simple.
    A pictorial representation of our recursive protocol is given in Figure
    \ref{fig:clocklayers}, and a pseudocode is given in Algorithm \ref{alg:intermediate}%
    \footnote{The pseudocode deviates from the presentation done in the proof,
    as it makes no use of recursion.}.

    Let us consider what we obtain after applying \corelemma
    the 
    first time to \protocol $=$\simple for clocks modulo $T$. 
    %In this case \simple uses messages of length $\numbits = \log T$.
    Recall that we assume that $T$ is a power of 2. 
    From Proposition \ref{prop:simple} we know that in this case, the
    convergence time of \simple is $\runningtime = \bigO\left( \log T \log
    n\right)$, the number of pulled agents at each round is $2$ and
    the number of bits of each message is $\numbits = \log T$. 

    With the emulation produced by \corelemma, the \counter used in \protocol =\simple is incremented only every $\numbits =\log T$ rounds. 
    %If we consider the sequence of times given by the \counter, we obtain $0,0,\ldots,1,1, \ldots,T,T\ldots T$ where each block is of length $\log T$, and this repeats. 
    Another way to interpret this is that we obtain a \counter modulo $T\cdot \numbits$ and using Remark \ref{rem:counters} we can view the clock $C:=T\cdot \numbits \mod T$,
    as a counter modulo $T$ that is incremented at each round. 
    Hence, by the running time analysis of \corelemma, we obtain a protocol $\emul$ which synchronizes a \counter modulo $T$ in
    $\bigO\left(\log n\log \log T \right) + \bigO\left(\log^2 T \log n\right) = \bigO\left( \log^2 T \log n\right)$ rounds. 
    The message size is reduced from $\log T$ to $\lceil \log \log T \rceil +1 = \bigO\left(\log \log T \right)$.

    By repeatedly applying \corelemma, we reduce the size of the message $\numbits$ as long as $\numbits > \lceil \log \numbits \rceil +1$, i.e. as long as $\numbits>3$. 
    The number of repeated application of \corelemma until the message size is $3$ is thus of order $\log^{*} T$. 

    \begin{pseudocode}[!ht]
        \textbf{\intermediate protocol}

        \smallskip
        \textsc{Memory:}
        Each agent $u$ keeps a sequence of \counters $\clock_1,\ldots,
        \clock_\finaliter$ and a sequence of bits $b_1, \ldots,
        b_\finaliter$. The clock $\clock_1$ runs modulo $T$, the \counter
        $\clock_\finaliter$ runs modulo $4$, and the $i$-th \counter
        $\clock_i$ runs modulo $2^{\ell_i-1}$ (see proof of Lemma
        \ref{lem:threebits}).
        Each agent $u$ also maintains a sequence of heaps (or some ordered
        structure) $S^{\delta}_i$, for each $\delta \in \{1, 2\}$ and $i =
        1,\ldots, \finaliter$.

        \smallskip
        \textsc{Message:} 
        $u$ displays $C_\finaliter$ ($2$ bits) and $b_\finaliter$ ($1$
        bit). For all $i \in [\finaliter], b_i(u)$ is the
        $\clock_i(u)$-th bit of the string obtained concatenating the
        binary representation of $C_{i-1}(u)$ and $b_{i-1}(u)$.

        \smallskip
        \begin{description}
            \item[1] $u$ samples two agents $u_1$ and $u_2$.
            \item[2] $u$ updates its clock with the bitwise majority of its clock and those
                of the sampled nodes.
            \item[3] $u$ increments its clock by one unit.
            \item[4] $u$ sets $i^*$ equal to the maximal $i< \finaliter$
                such that $\clock_{i+1} \neq 0$.
            \item[5] For $\delta=1,2$, $u$ pushes $b_\finaliter(u_\delta)$ in $S^{\delta}_{i^*}$.\\
                (Note that, if $\clock_{i^*+1}, \ldots, \clock_{\finaliter}$ are
                synchronized, then all agents are displaying the bit with index
                $\clock_{i^*+1}$ of $(\clock_{i^*},b_{i^*})$ as $b_{\finaliter}$.)
            \item[6] {While} $i>1$ and $\clock_{i} = 0$, $u$ does the following:
                \item[7] \homemadeloop Pops the last $\ell_{i-1}-1$ bits from $S^\delta_{i-1}$ and 
                    set $s^\delta$ equal to it. 
                \item[8] \homemadeloop Sets $\clock_{i-1}$ equal to the bitwise majority of
                    $\clock_{i-1}(u)$, $s^{1}$ and $s^{2}$.
                \item[9] \homemadeloop Increments $\clock_{i-1}$ and
                    decrement $i$ by one unit.
        \end{description}
        \caption{Iterative version of the protocol \intermediate,
            executed by each agent $u$, unfolding the recursion in proof of
            Lemma \ref{lem:threebits}.}
        \label{alg:intermediate}
    \end{pseudocode}

    Let us analyze the running time. 
    Let $\numbits_1 = \log T$, $\numbits_{i+1} =  \lceil \log \numbits_i \rceil +1$
    and let $\finaliter(T) = \finaliter$ be the smallest integer such that $\numbits_\finaliter= 3$.
    We apply \corelemma $i\leq \finaliter$ times, and we obtain a message size $\numbits_i$ and a running time $L_i$, such that 
    \begin{align}
        L_{i+1} \leq \const_1(\log \numbits_i \log n + \numbits_i L_i),
        \label{eq:recruntime}
    \end{align}
    for some constant $\const_1$ independent of $i$.
    We set $L_1$ to be
     $L_1 :=  L_{\simple} \vee \log n =  \bigO\left(\log T \log n\right) \vee \log n$, taking the maximum with $\log n$ for technical convenience. 
    The second term dominates in (Equation \ref{eq:recruntime}) because $\numbits_i >> \log \numbits_i$ and $L_i > \log n$. Hence $L_i$ is at most of order $\prod_{j<i} \numbits_j\cdot L_1$. More precisely, by induction we can bound $L_i \leq \const_1^i\prod_{j=1}^{i-1} \numbits_{j} L_1$, since
    \begin{align}
        L_{i+1} 
        \leq \const_1\log \numbits_i \log n + \const_1^i\prod_{j=1}^{i} \numbits_{j} \cdot L_1
        \leq \const_1 \numbits_i \log n + \const_1^i \prod_{j=1}^{i} \numbits_j \cdot L_1 
        \leq 2\const_1^i \prod_{j=1}^{i} \numbits_j \cdot L_1 
            \leq \const_1^{i+1} \prod_{j=1}^{i} \numbits_j \cdot L_1,
    \end{align}
    where we use the fact that $\const_1 >2$, 
    and the definition of~$L_1$.

    The running time of \emul$=$\synclock after the last application of \corelemma, i.e. $\finaliter$, is thus
    \begin{align}
        \runningtimethree := L_{\finaliter} \leq \const_1^{\finaliter} 
        \prod_{i=1}^{\finaliter} \numbits_i L_1.
    \end{align}
    %\iftoggle{full}
    We use the following fact.
    %{Next we use the following, whose proof is provided in \cite{BKN16arxiv}.}
    %
    \begin{fact}
        \label{fact:exp_ineq}
        If $|x|<1$, it holds
        \[
        e^{\frac{x}{1+x}}\leq1+x\leq e^{x}\leq 1+\frac{x}{1-x}.
        \]
    \end{fact}
    From the bounds 
    $L_1 = \bigO(\log T \log n)$, 
    $\prod_{i=1}^{\finaliter} \numbits_i \leq \numbits_1 \numbits_2 \numbits_3^{\finaliter}$, 
    $\numbits_1 = \bigO\left(\log T\right)$, 
    $\numbits_2 = \bigO\left(\log \log T\right)$ and Lemma \ref{lem:gt}, we obtain 
    $\const_1^{\finaliter} = \bigO\left(\log \log \log T\right)$ and
    \[
        \numbits_3^{\finaliter} \leq 2^{\bigO\left((\log^{\circledast 4} T)^2\right)}
        \leq 2^{\bigO\left( \log \log \log T\right)}
        \leq \left(\log \log T \right)^{\bigO(1)}.
    \]
    We thus conclude that
    \begin{align}
        \runningtimethree 
         &\leq \const_1^{\finaliter} \prod_{i=1}^{\finaliter} \numbits_i L_1 
        \leq \bigO\left(\log \log \log T\right) \cdot \numbits_1 \numbits_2 \numbits_3^{\finaliter} \cdot \bigO(\log T \log n)\\
        &\leq  \bigO\left(\log \log \log T\right) 
            \cdot \bigO\left(\log T\right) 
            \cdot \bigO\left(\log \log T\right) 
            \cdot \bigO\left(\log \log T \right)^{\bigO(1)} 
            \cdot \bigO(\log T \log n)\\
        &\leq  \log^2 T \log n \cdot  \left(\log \log T\right)^{\bigO(1)}.
    \end{align}
    The total slowdown with respect to \simple corresponds to 
    $\prod_{i=1}^\tau \numbits_i = \tilde{\bigO}(\log T)$. Hence the \counter produced by the emulation is incremented every $\tilde{\bigO}(\log T)$ rounds.
    In other words we obtain a \counter modulo $T\cdot f(T)$ for some function $f$.
    But using Remark \ref{rem:counters} we can still view this as a \counter modulo $T$.
\qed

\subsection{Extension to general $T$ and running time improvement.}\label{sub:powerof}

In this subsection we aim to get rid of the assumption that $T$ is a power of $2$ in Lemma \ref{lem:threebits}, and also reduce the running time of our protocol to $\tilde{\bigO}\left(\log n \log T\right)$, proving Theorem \ref{thm:main}.

\begin{pseudocode}[!ht]
    \textbf{\synclock protocol}

    \smallskip
    \textsc{Memory:}
    Each agent $u$ stores a \counter $C'(u)$ which runs modulo $T' \gg \const
    \log n \log T$. Each agent $u$ also stores a variable $Q$ which is
    incremented only once every $T'$ rounds and runs modulo $T$.

    \smallskip
    \textsc{Message:}
    Each agent $u$ displays $4$ bits.
    On the first $3$ bits, protocol \intermediate is applied to synchronize $C'$. 
    The $4$-th bit $b(u)$ is the
    bit with index $(\lfloor \frac{C'(u)}{\const\log n} \rfloor \mod \lceil
    \log T \rceil )$ of $Q(u)$.

    \smallskip
    \begin{description}
        \item[1] $u$ samples two agents $u_1$ and $u_2$.
        \item[2] $u$ updates $b(u)$ with the majority of $b(u)$, $b(u_1)$ and $b(u_2)$.
        \item[3] If $C'=0$, increment $Q$ by one unit modulo $T$.
    \end{description}

    \smallskip
    \textsc{Output:} 
    The \counter modulo $T$ is obtained as $ C := \left(C'+ Q \cdot T'\right) \mod T$
    \caption{The protocol $4$-bit \synclock, executed by each agent $u$.}
    \label{alg:synclock}
\end{pseudocode}

\proofof{Theorem \ref{thm:main}}
    From Lemma \ref{lem:threebits}, we know that 
    \intermediate synchronizes \counters modulo $T$ in time $ \tilde{\bigO}\left(\log^2 T \log n\right)$ using only $3$-bit messages, provided that $T$ is a power of 2. 
    While protocol \intermediate emulates protocol \simple, it displays the first bit of the message of \simple only once every $\tilde{\bigO}\left(\log T\right)$ rounds. Of course, it would be more efficient to display it $\bigO\left(\log n\right)$ times in a row, so that \major would make every agent agree on this bit, and then move to agreeing on the second bit, and so on.
    To achieve this, as in the proof of \simple, we can view a \counter modulo $T$, say $Q$, as written on $\log T $ bits. If agents already possess a ``small'' counter modulo $T' := \bigO\left(\log T \log n\right)$ they can use it to display the first bit for $\bigO\left(\log n\right)$ rounds, then the second one for  $\bigO\left(\log n\right)$ rounds, and so on until each one of the $\lceil \log T \rceil$ bits of $T$ has been synchronized. 
    This would synchronize all bits of the desired \counter within $\bigO\left(\log T \log n\right)$ rounds, w.h.p., while being very economical in terms of message length, since only $1$ bit is displayed at any time.

    Therefore,  we can  use Lemma \ref{lem:threebits} to synchronize a counter modulo 
    $\bigO\left(\log T \log n \right)$ in  $\tilde \bigO ((\log \log T)^2\log n )$ rounds,
    using $3$ bits per message. Then, we can use a fourth bit to run \major on each of the $\log T$ bits of $Q$ for $\bigO(\log n)$ consecutive rounds, for a total running time of $\bigO(\log T \log n)$ rounds. At this point, an application of \corelemma would give us a protocol with running time $\bigO(\log T \log n)$ using $3$-bit messages. However, perhaps surprisingly, a similar strategy enables us to synchronize a clock modulo any integer (not necessarily a power of $2$). 
    %
    % Recall the fact that $T$ is a power of $2$ was required at the end of the proof of Proposition \ref{prop:simple}, so that, once agents agree on the first bit, the increment of time makes them flip the second bit \emph{jointly} once every $2$ rounds and, more generally, if agents agree on the first $i$ bits of their clocks then they \emph{jointly} flip the $i +1$'st bit once every $2^{i}$ rounds (while they apply the $\major$ protocol on that bit, in parallel).
    % In other words, the assumption that $T$ is a power of $2$ was used to ensure that we can progressively synchronize the bits.
    % However, it turns out we can bypass the aforementioned issue by updating $Q$ only once in a while, as we explain in the following. 

    Let us assume that $T\in \NN$ is an arbitrary integer. Let $\consconvtime$ be an upper bound on the convergence time of \major which guarantees a correct consensus with probability at least $1-n^{-\consconvprob}$, for some constant $\const$ large enough \cite{DGM11}. Let $T'$ be the smallest power of $2$ bigger than $\log T \cdot \left(\agreementQ\right)$.
    By Lemma \ref{lem:threebits}, using $3$ bits, the agents can build a synchronized \counter $C'$ running modulo $T'$ in time $\tilde \bigO ((\log \log T)^2\log n )$. The other main ingredient in this construction is another \counter $Q_{T'}$ which is incremented once every $T'$ rounds and runs modulo $T$. 
    The desired clock modulo $T$, which we denote $C$, is obtained by
    \begin{align}
        C := \left(C' + Q_{T'} \cdot T' \right)\mod T.
    \end{align}
    It is easy to check, given the definitions of $C'$ and $Q_{T'}$ that this choice indeed produces a clock modulo $T$.
    %Indeed, if $Q$ was unbounded, $C' + Q \cdot T'$ counts to $\infty$ (we are simply writing the Euclidean division of a number by $T'$). 
    %When taking $Q$ modulo $T$, we obtain a counter modulo $T\cdot T'$. As pointed out in Remark \ref{rem:counters}, when we then take this expression modulo $T$, we obtain a proper counter modulo $T$.

    It remains to show how the \counter $Q_{T'}$ modulo $T$ is synchronized. 
    On a first glance, it may seem as if we did not simplify the problem since $Q$ is a clock modulo $T$ itself. However, the difference between $Q_{T'}$ and a regular \counter modulo $T$ is that $Q_{T'}$ is incremented only once every $T'$ rounds. This is exploited as follows.

    The counter $Q_{T'}$ is written on $\lceil \log T\rceil$ internal bits. We show how to synchronize $Q_{T'}$ using a 4-th bit in the messages, similarly to the aforementioned strategy to synchronize $Q$; we later show how to remove this assumption using \corelemma. Let us call a loop of $C'$ modulo $T'$ an \emph{epoch}. The rounds of an epoch are divided in phases of equal length $\agreementQ$ (the remaining $T' \mod (\agreementQ)$ rounds are just ignored). 
    The clock $C'$ determines which bit from $Q_{T'}$ to display. The first bit of $Q_{T'}$ is displayed during the first phase, then the second one is displayed during the second phase, and so on. By Theorem \ref{thm:doerr}, the length of each phase guarantees that consensus is achieved on each bit of $Q_{T'}$ via%
    \footnote{Observe that, once \clockname $C'$ is synchronized, the bits of
        $Q_{T'}$ do not change for each agent during each subphase. Thus, we
        may replace \major by the \textsc{Min} protocol where on each round of
        subphase $i$ each agent $u$ pulls another agent $v$ u.a.r. and updates her
        $i$-th bit of $Q$ to the minimum between her current $i$-th bit of $Q$ and the
        one of $v$. However, for simplicity's sake, we reuse the already introduced
        \major protocol.}
    \major w.h.p. 
    More precisely, after the first bit has been displayed for $\agreementQ$ rounds, all agents agree on it with probability%
    \footnote{From Theorem \ref{thm:doerr}, we have that after $\const \log n$ rounds, with $\const$ large enough, the probability that consensus has not been reached is smaller than $ \frac{1}{n^2}$. Thus, after $\consttwo \cdot \const \log n$ rounds, the probability that consensus has not been reached is smaller than $\frac{1}{n^{2\consttwo}}$. 
    If we choose $\consttwo \log n = \log n + \log \log T$, we thus get the claimed upper bound $\frac{1}{n^2 \log T}$.}
    $1 - \frac{1}{n^2 \log T}$, provided $\constpower$ is large enough. 
    Thus, at the end of an epoch, agents agree on all $\lceil \log T \rceil$ bits of $Q_{T'}$ with probability greater than $(1-\frac{1}{n^2 \log T})^{\log T}\gg 1-\bigO(n^{-2})$.

    We have thus shown that, by the time $C'$ reaches its maximum value of $T'$, i.e. after one epoch, all agents agree on $Q_{T'}$ w.h.p. and then increment it jointly.
    From Lemma \ref{lem:threebits}, \intermediate takes $\tilde{\bigO}\left(\log^2 T' \log n\right) = \bigO\left((\log \log n + \log \log T)^2 \log n\right) = \bigO\left(\left((\log \log n)^2\log n + (\log \log T)^2 \log n\right)\right)$ rounds to synchronize a clock $C'$ modulo $T'$ w.h.p.
    Together with the $\log T\left(\agreementQ\right)$ rounds to agree on $Q_{T'}$ w.h.p., this implies that after $\log T \log n \cdot
    (\log \log T)^{\bigO(1)} \cdot (\log \log n)^{\bigO(1)}=\tilde{\bigO}\left(\log T \log n\right)$ rounds the clocks $C$ are all synchronized w.h.p.

    Finally, we show how to get rid of the extra $4$-th bit to achieve agreement on $Q_{T'}$. Observe that, once $\clock'$ is synchronized, this bit is used in a self-stabilizing way. Thus, since \intermediate has the \independence property, using Lemma \ref{lem:ipfact}, the protocol we described above possesses the \independence property too. 
    By using \corelemma we can thus reduce the message size from $4$ bits to $\lceil \log 4 \rceil +1 = 3$ bits, while only incurring a constant multiplicative loss in the running time. 
    The  \clockname we obtain, counts modulo $T$ but is incremented every $4$ rounds only. It follows from Remark \ref{rem:counters} that we may still view this as a \counter modulo $T$.
\qed

\begin{remark}[Internal memory space]
The internal memory space needed to implement our protocols \simple, \intermediate, and \synclock is close to $\log T$ in all cases:
protocol \simple uses one counter written on $\log T$ bits, \intermediate needs internal memory of size 
\begin{equation}
   \log T + \bigO\left(\log \log T+ \log \log \log T + \ldots\right) \leq  \log T (1+o(1)),
\end{equation}
and the internal memory requirement of  \synclock is of order $\log T + \log \log n$. 
\end{remark}

\section{\majority \broadcast with a Clock}\label{sec:majority}

In this section we assume that agents are equipped with a synchronized \counter $\clock$ modulo $\const \log n$ for some big enough constant $\const>0$. In the previous section we showed how to establish such a synchronized clock in $\tilde{\bigO}(\log n)$ time and using 3-bit messages.
We have already seen in Section \ref{sec:technical} how to solve the \broadcast problem (when we are promised to have a single source agent) assuming such synchronized \counters, by paying an extra bit in the message size and an $\bigO(\log n)$ additive factor in the running time. 
This section is dedicated to showing that, in fact, the  more general \majority \broadcast problem can be solved with the same time complexity and using 3-bit messages, proving Theorem \ref{thm:multi}.

In Section \ref{sec:spread}, we describe and analyze protocol \majprot, which solves
\majority \broadcast
by paying only a $\bigO(\log n)$ additive overhead in the running time w.r.t. \clocksyncronization. 
For clarity's sake, we first assume that the protocol is using $4$ bits (i.e. 1 additional bit over the $3$ bits used for \clocksyncronization), and we later show how to decrease the number of bits back to 3 in Section \ref{sub:proofofmulti}, by applying \corelemma. 

The main idea behind the $3(+1)$-bit protocol, called
\majprot, is to make the sources' \inputbits disseminate on the system in a way
that preserves the initial ratio $\frac{\sizeblack}{\sizewhite}$ between the
number of sources supporting the majority and minority \inputbit. This is
achieved by dividing the dissemination process in phases, similarly to the main
protocol in \cite{OHK14} which was designed to solve the \broadcast problem in
a variant of the $\push$ model in which messages are affected by noise. The
phases induces a spreading process which allows to leverage on the
concentration property of the Chernoff bounds, preserving the aforementioned
ratio. While, on an intuitive level, the role of noisy messages in the model
considered in \cite{OHK14} may be related to the presence of sources having
conflicting opinion in our setting, we remark that our protocol and its
analysis depart from those of \cite{OHK14} on several key points: while the
protocol in \cite{OHK14} needs to know the the noise parameter, \majprot do not
assume any knowledge about the number of different sources, and our analysis do
not require to control the growth of the number of speaking agents from above%
\footnote{To get such upper bound, the analysis in \cite{OHK14} leveraged on the property that in the $\push$ model the number of agents getting a certain message can be upper bounded by the number of agent sending such message, which is not the case for the passive communication of the $\pull$ model.}.

In order to perform such spreading process with 1 bit only, the protocol in \cite{OHK14} leverages on the fact that in the $\push$ model agents can choose {\em when to speak}, i.e. whether to send a message or not. To emulate
this property in the $\pull$ model, we use the parity of the clock $\clock$: on
odd rounds agents willing to ``send'' a  $0$ display $0$, while others display $1$ and conversely on even rounds. Rounds are then grouped by two, so $2$ rounds in the $\pull$ model correspond to $1$ round in the $\push$ version.

\subsection{Protocol \majprot}\label{sec:spread}
 
\global\long\def\prot{\textsc{Com-Spread}}

\global\long\def\constclock{\gamma_{phase}}
\global\long\def\clock{C}
\global\long\def\timeprot{2\lceil \newconstboost \log n\rceil + \newconstboost\lceil 2\log n \rceil}

\global\long\def\constphase{\gamma_{phase}}
\global\long\def\constboost{\constphase}
\global\long\def\newconstboost{\gamma_{phase}}

\global\long\def\consttrial{\gamma_{trial}}

In this section we describe protocol \majprot. As mentioned above, for clarity's sake we assume that \majprot uses $4$-bit messages, and we show how to remove this assumption in Section \ref{sub:proofofmulti}.
Three of such bits are devoted to the execution $\synclock$, in order to synchronize a \counter $\clock$ modulo $\timeprot$
for some constant $\constclock$ large enough.
Throughout this section we assume, thanks to Theorem \ref{thm:main}, that $\clock$ has already been synchronized, which happens after $\tilde{\bigO} (\log n)$ rounds from the start of the protocol. In Section \ref{sec:majprot}, we present a protocol \majprotgiven solving \majority \broadcast assuming agents already share a common clock.

\subsubsection{Protocol \majprotgiven}
\label{sec:majprot}
Let $\newconstboost$ be a constant to be set later. 
Protocol
\majprotgiven is executed periodically over periods of length $\timeprot$, given by a \counter $\clock$. One run of length $\timeprot$ is divided in $2+\lceil 2\log n \rceil$ phases, the first
and the last ones lasting $\lceil \newconstboost\log n \rceil$ rounds, all the other
$\lceil 2\log n \rceil$ phases lasting $\newconstboost$ rounds. The first phase is called \emph{boosting},
the last one is called \emph{polling}, and all the intermediate ones
are called \emph{spreading}.
For technical convenience, in \majprotgiven agents disregard the messages they get as their second pull\footnote{In other words, \majprotgiven works in the $\pull(1)$ model.}.

During the boosting and the spreading phases, as we already explained in the introduction of this section, we make use of the parity of time to emulate the ability to actively send a message or not to communicate anything as in the $\push$ model%
\footnote{Of course, agents are still not able to control who sees/contacts them.}
(in the first case we say that the agent is \emph{speaking}, in the second case we say that the agent is \emph{silent}).
This induces a factor $2$ slowdown which we henceforth  omit for simplicity.

\global\long\def\freqblack{f_{1}}
\global\long\def\freqwhite{f_{0}}

At the beginning of the boosting, each non-source agent $u$ is
silent. 
During the boosting and during
each spreading phase, each silent agent pulls until she sees a speaking
agent. When a silent agent $u$ sees a speaking agent $v$,
$u$ memorizes $\bitfirst\left(v\right)$ but remains silent until
the end of the phase; at the end of the current phase, $u$ starts speaking and sets $\bitfirst\left(u\right)=\bitfirst\left(v\right)$.
The bit $\bitfirst$ is then never modified until the \counter $\clock$ reaches $0$ again.
Then, during the
polling phase, each agent $u$ counts how many agents with $\bitfirst=1$
and how many with $\bitfirst=0$ she sees.
At the end of the phase, each agent $u$ sets their output bit to the most frequent value of $\bitfirst$ observed during the
polling phase.
We want to show that, for all agents, the latter is w.h.p. $\bitmaj$ (the most frequent initial \outputbit among sources).

\begin{pseudocode}[!ht]
    \textbf{\majprotgiven protocol}

    \smallskip
    \begin{description}
        \item[1] If $u$ is not speaking and the current phase is either the boosting or 
            the spreading one, $u$ does the following:
            \item[2] \homemadeloop $u$ observes a random agent $v$. 
            \item[3] \homemadeloop If $v$ is speaking, 
                $u$ sets $\bitfirst(u)$ equal to $\bitfirst(v)$,\\
                \-\hspace{.5cm} and $u$ will be speaking from the next phase.
            \item[4] \homemadeloop $u$ sets $c_0$ and $c_1$ equal to $0$.
        \item[5] If the current phase is polling:
            \item[6] \homemadeloop $u$ observes a random agent $v$. 
            \item[7] \homemadeloop If $b_1(v)=1$, $u$ increments $c_1$, 
                otherwise increment $c_0$. 
        \item[8] $u$ outputs $1$ if and only if $c_1 > c_0$.
    \end{description}
    \caption{The protocol \majprotgiven, executed by each agent $u$.}
    \label{alg:majprot}
\end{pseudocode}

\subsubsection{Analysis}
We prove that at the end of the last spreading phase w.h.p.
all agents are speaking and each agent has $\bitfirst=1$ with probability
$\frac{1}{2}+\trialbias$ for some positive constant $\trialbias=\trialbias\left(\newconstboost, \epsilon\right)$ (where the dependency in $\constphase$ is monotonically increasing),
$\bitfirst=0$ otherwise. From the Chernoff bound (Corollary \ref{cor:cb})
and the union bound, this implies that when $\newconstboost>\frac{8}{\trialbias}$
at the end of the polling phase w.h.p. each agent learns
$\bitmaj$. 

Without loss of generality, let $\bitmaj=1$, i.e. $\sizeblack>\sizewhite$. 

The analysis is divided in the following lemmas.
\begin{lemma}
    \label{lem:boosting_conc}
    At the end of the boosting phase it holds
    w.h.p.
    \begin{align}
        \boostblack+\boostwhite 
        & \geq\left(\sizeblack+\sizewhite\right)\frac {\constphase}3 \log
                n\cdot\mbox{\textbf{1}}_{\left\{
                    \sizeblack+\sizewhite<\frac{n}{2\newconstboost \log n}\right\} }
            +\left(n\left(1-\frac{1}{\sqrt{e}}\right) 
            +\frac{1}{\sqrt{e}}\left(\sizeblack+\sizewhite\right)-\sqrt{n\log
                n}\right)\nonumber \\
         & \qquad \cdot\mbox{\textbf{1}}_{\left\{
             \frac{n}{2\newconstboost}\leq\sizeblack+\sizewhite\leq n-2\sqrt{n\log
             n}\right\} }
         \qquad +n\mbox{\textbf{1}}_{\left\{
                 \sizeblack+\sizewhite>n-2\sqrt{n\log n}\right\}
             },\label{eq:boost_spread}\\
        \frac{\boostblack}{\boostwhite} 
             &\geq\frac{\sizeblack}{\sizewhite}\left(1-\sqrt{\frac{9}{\newconstboost\sizewhite}}\right).
        \label{eq:boost_bias}
    \end{align}
    \global\long\def\E{\mathbb{E}}
\end{lemma}

\begin{proof}%[Proof of Lemma \ref{lem:boosting_conc}]
    First, we prove \eqref{eq:boost_spread}. By using Fact \ref{fact:exp_ineq},
    we have
    \begin{align}
        \E\left[\boostblack+\sizewhite^{\left(1\right)}\right]
        & =\sizeblack+\sizewhite+\left(n-\sizeblack-\sizewhite\right)
            \left(1-\left(1-\frac{\sizeblack+\sizewhite}{n}\right)^{\newconstboost
            \log n}\right)\nonumber \\
        & \qquad\geq\sizeblack+\sizewhite+\left(n-\sizeblack-\sizewhite\right)
            \left(1
                -e^{-\frac{\sizeblack+\sizewhite}{n}\newconstboost \log n}\right).
        \label{eq:boost_before_cases}
    \end{align}
    We distinguish three cases.

    \textbf{Case $\sizeblack+\sizewhite<\frac{n}{2\newconstboost\log n}$.}
    By using Fact \ref{fact:exp_ineq} again, from \eqref{eq:boost_before_cases}
    we get
    \begin{align}
        \E\left[\boostblack+\sizewhite^{\left(1\right)}\right] 
        & \geq\sizeblack+\sizewhite+\left(n-\sizeblack-\sizewhite\right)
            \left(1-e^{-\frac{\sizeblack+\sizewhite}{n}\newconstboost\log
            n}\right)\nonumber \\
        & \geq\sizeblack+\sizewhite+\left(n-\sizeblack-\sizewhite\right)
            \frac{\frac{\sizeblack+\sizewhite}{n}\newconstboost\log n}
                {1+\frac{\sizeblack+\sizewhite}{n}\newconstboost\log n}\nonumber \\
        & \geq\sizeblack+\sizewhite+
            \left(n-\sizeblack-\sizewhite\right)
            {\frac{\sizeblack+\sizewhite}{n}\frac{\newconstboost}2 \log n}\nonumber \\
        & \geq\sizeblack+\sizewhite+\left(1-\frac{\sizeblack+\sizewhite}{2n}\right)
            \left(\sizeblack+\sizewhite\right)\frac{\newconstboost}2
        \log n\nonumber\\
        &\geq 
          \left(\sizeblack+\sizewhite\right)
          \left(1+\left(1-\frac{1}{4\newconstboost\log n}\right)\frac{\newconstboost}2\log n\right)\\
        &\geq \left(\sizeblack+\sizewhite\right)\frac{\newconstboost}2 \log n.
        \label{eq:exp_spread}
    \end{align}
    From the Chernoff bound (Lemma \ref{thm:gen_cb}), we thus get that
    w.h.p. 
    \[
    \boostblack+\sizewhite^{\left(1\right)}\geq  \left(\sizeblack+\sizewhite\right)\frac{\newconstboost}{3} \log n.
    \]

    \textbf{Case $\frac{n}{2\newconstboost\log n} \leq \sizeblack+\sizewhite\leq n-2\sqrt{n\log n}$.}
    From \eqref{eq:boost_before_cases}, we have
    \begin{align}
        \E\left[\boostblack+\sizewhite^{\left(1\right)}\right] 
        & \geq\sizeblack+\sizewhite+\left(n-\sizeblack-\sizewhite\right)
            \left(1-e^{-\frac{\sizeblack+\sizewhite}{n}\newconstboost}\right)\\
        & \geq\sizeblack+\sizewhite+\left(n-\sizeblack-\sizewhite\right)
            \left(1-\frac{1}{\sqrt{e}}\right)\\
        &\geq n\left(1-\frac{1}{\sqrt{e}}\right)+\frac{\sizeblack+\sizewhite}{\sqrt{e}}.
    \end{align}
    From the Chernoff bound (Lemma \ref{thm:gen_cb}), we thus get that
    w.h.p. 
    \[
    \boostblack+\sizewhite^{\left(1\right)}\geq n\left(1-\frac{1}{\sqrt{e}}\right)+\frac{\sizeblack+\sizewhite}{\sqrt{e}}-\sqrt{n\log n}.
    \]

    \textbf{Case $\boostblack+\boostwhite>n-2\sqrt{n\log n}$.} The probability
    that a silent agent does not observe a speaking one is 
    \[
    \left(\frac{n-\sizeblack-\sizewhite}{n}\right)^{\newconstboost \log n}\leq 
    \left(\frac{4\log n}{n}\right)^{\frac{1}{2}\newconstboost \log n},
    \]
    hence by a simple union bound it follows that w.h.p.
    all agents are speaking. 

    Now, we prove \eqref{eq:boost_bias}. As before, we have two cases.
    The first case, $\frac{\sizeblack}{\sizewhite}\geq\frac{n}{2\newconstboost\log n}$,
    is a simple consequence of the Chernoff bound (Lemma \ref{thm:gen_cb}). 

    \global\long\def\setboost{S_{boost}}

    In the second case, $\frac{\sizeblack}{\sizewhite}<\frac{n}{2\newconstboost\log n}$,
    let us consider the set of agents $\setboost$ that start speaking
    at the end of the boosting, i.e. that observe a speaking agent during
    the phase. Observe that $\left|\setboost\right|=\boostblack-\sizeblack+\sizewhite^{\left(1\right)}-\sizewhite$.
    The probability that an agent in $\setboost$ observes an agent in
    $\typeblack$ (resp. $\typewhite$) is $\frac{\sizeblack}{\sizeblack+\sizewhite}$
    (resp. $\frac{\sizewhite}{\sizeblack+\sizewhite}$). Thus
    \begin{align}
    \E\left[\boostblack\right] & =\sizeblack+\frac{\sizeblack}{\sizeblack+\sizewhite}\E\left[\left|\setboost\right|\right]\quad\mbox{and}\nonumber \\
    \E\left[\sizewhite^{\left(1\right)}\right] & =\sizewhite+\frac{\sizewhite}{\sizeblack+\sizewhite}\E\left[\left|\setboost\right|\right].\label{eq:exp_white_boost}
    \end{align}
    In particular 
    \begin{equation}
      \frac{\E\left[\boostblack\right]}{\E\left[\sizewhite^{\left(1\right)}\right]}=\frac{\sizeblack+\frac{\sizeblack}{\sizeblack+\sizewhite}\E\left[\left|\setboost\right|\right]}{\sizewhite+\frac{\sizewhite}{\sizeblack+\sizewhite}\E\left[\left|\setboost\right|\right]}=\frac{\sizeblack}{\sizewhite},\label{eq:exp-ratio}  
    \end{equation}
    and from \eqref{eq:exp_spread} and \eqref{eq:exp_white_boost} we have
    \begin{align}
        \E\left[\sizewhite^{\left(1\right)}\right]
        &\geq\frac{\sizewhite}{\sizeblack+\sizewhite}\E\left[\left|\setboost\right|\right]\\
        & =\frac{\sizewhite}{\sizeblack+\sizewhite}
            \left(\E\left[\boostblack+\sizewhite^{\left(1\right)}\right]
                -\left(\sizeblack+\sizewhite\right)\right)\\
        &\geq (1-o(1)) \frac{\sizewhite}{\sizeblack+\sizewhite}\frac{\constphase}{2}\left(\sizeblack+\sizewhite\right)\log n\\
        &= (1-o(1)) {\sizewhite}\frac{\constphase}{2} \log n,
        \label{eq:exp_white_boost_lower}
    \end{align}
    where the lower bound follows from the assumption $\frac{\sizeblack}{\sizewhite}<\frac{n}{2\constphase\log n}$ and
    \eqref{eq:exp_spread}.
    From \eqref{eq:exp_white_boost_lower} and the multiplicative form of the Chernoff bound (Corollary \ref{cor:cb}),
    we have that w.h.p. 
    \begin{align}
        & \boostblack  \geq\E\left[\boostblack\right]-\sqrt{\E\left[\boostblack\right]\log n}\quad\mbox{and}\\
        & \boostwhite  \leq\E\left[\boostwhite\right]+\sqrt{\E\left[\boostwhite\right]\log n}.\label{eq:exp_white_boost_whp}
    \end{align}

    Thus, since \eqref{eq:exp_white_boost} implies $\E\left[\boostblack\right]\geq \E\left[\boostwhite\right]$, we have
    \begin{align}
         \frac{\boostblack}{\boostwhite}  
         &   \geq\frac{\E\left[\boostblack\right]-\sqrt{\E\left[\boostblack\right]\log n}}
                {\E\left[\boostwhite\right]+\sqrt{\E\left[\boostwhite\right]\log n}}\\
        & =\frac{\E\left[\boostblack\right]}{\E\left[\boostwhite\right]}
            \cdot \frac{1-\sqrt{\frac{\log n}{\E\left[\boostblack\right]}}}
                {1+\sqrt{\frac{\log n}{\E\left[\boostwhite\right]}}}\\
        &\geq\frac{\E\left[\boostblack\right]}{\E\left[\boostwhite\right]}
            \cdot \left(1-\sqrt{\frac{\log n}{\E\left[\boostblack\right]}}
                -\sqrt{\frac{\log n}{\E\left[\boostwhite\right]}}\right) \\
        & \geq\frac{\E\left[\boostblack\right]}{\E\left[\boostwhite\right]}
            \cdot \left(1-2\sqrt{\frac{\log n}{\E\left[\boostwhite\right]}}\right) \\
        & =\frac{\sizeblack}{\sizewhite}
            \cdot\left(1-\sqrt{\frac{9}{\sizewhite\newconstboost}}\right),
        \label{eq:recurs_whp}
    \end{align}
    concluding the proof.
\end{proof}
\begin{lemma}
    At the end of the $i+1$th spreading phase, the following holds w.h.p.
    \begin{align}
        \nextblack+\nextwhite
        & \geq\left(\nowblack+\nowwhite\right) \frac{\constboost}3 \mbox{\textbf{1}}_{\left\{ \nowblack+\nowwhite<\frac{n}{2\newconstboost}\right\} }
        +\left(n\left(1-\frac{1}{\sqrt{e}}\right)+\frac{1}{\sqrt{e}}
            \left(\nowblack+\nowwhite\right)-\sqrt{n\log n}\right)\nonumber \\
        & \qquad \qquad \cdot\mbox{\textbf{1}}_{\left\{ \frac{n}{2\newconstboost}
            \leq\nowblack+\nowwhite\leq n-2\sqrt{n\log n}\right\} }
        +n\mbox{\textbf{1}}_{\left\{ \nowblack+\nowwhite>n-2\sqrt{n\log n}\right\} },
        \label{eq:increase_big}\\
         \frac{\nextblack}{\nextwhite}
         &   \geq\frac{\nowblack}{\nowwhite}\left(1- 4 \sqrt{\frac{\log n}
                {\newconstboost\nowwhite}}\right).
        \label{eq:rel_bias}
    \end{align}
\end{lemma}
\begin{proof}
    The proof is almost the same as that of Lemma \ref{lem:boosting_conc}. Thus, 
    we here condense some analogous calculations.

    First, we prove \eqref{eq:increase_big}. By using Fact \ref{fact:exp_ineq},
    we have
    \begin{align}
        \E\left[\nextblack+\nextwhite\right]
        %& =\nowblack+\nowwhite+\left(n-\nowblack-\nowwhite\right)
        %    \left(1-\left(1-\frac{\nowblack+\nowwhite}{n}\right)^{\newconstboost}\right)\\
         \geq\nowblack+\nowwhite+\left(n-\nowblack-\nowwhite\right)
            \left(1- e^{-\frac{\nowblack+\nowwhite}{n}\newconstboost}\right).
        \label{eq:boost_before_cases-1}
    \end{align}
    We distinguish three cases.

    \textbf{Case $\nowblack+\nowwhite<\frac{n}{2\newconstboost}$.} By using
    Fact \ref{fact:exp_ineq} again, from \eqref{eq:boost_before_cases-1}
    we get
    \begin{align}
        \E\left[\nextblack+\nextwhite\right]
        %& \geq\nowblack+\nowwhite+\left(n-\nowblack-\nowwhite\right)
        %    \cdot\left(1- e^{-\frac{\nowblack+\nowwhite}{n}\newconstboost}\right)\\
        %& \geq\nowblack+\nowwhite+\left(n-\nowblack-\nowwhite\right)\cdot
        %    \frac{\frac{\nowblack+\nowwhite}{n}\newconstboost}
        %        {1+\frac{\nowblack+\nowwhite}{n}\newconstboost}\\
        & \geq\nowblack+\nowwhite+\left(n-\nowblack-\nowwhite\right)\cdot
            \frac{\nowblack+\nowwhite}{2n}\newconstboost
        %& \geq\nowblack+\nowwhite+\left(1-\frac{\nowblack+\nowwhite}{n}\right)
        %    \cdot\left(\nowblack+\nowwhite\right)\frac{\newconstboost}2\\
        %& \geq \left(\nowblack+\nowwhite\right)
        %    \left(1+\left(1-\frac{1}{2\newconstboost}\right)\frac{\newconstboost}2\right)\\
        \geq \left(\nowblack+\nowwhite\right)\frac{\newconstboost}2.
        \label{eq:exp_spread-1}
    \end{align}
    After the boosting phase, i.e. for $i\geq 1$, it follows from Lemma \ref{lem:boosting_conc}
    that $\nowblack+\nowwhite = \Omega\left( \newconstboost \log n\right)$. From the Chernoff bound (Lemma \ref{thm:gen_cb}),
    if $\newconstboost$ is chosen big enough, we thus get that
    w.h.p. 
    \[
    \nextblack+\nextwhite\geq\left(\nowblack+\nowwhite\right)\frac \newconstboost 3.
    \]

    \textbf{Case $\frac{n}{2\newconstboost}\leq\nowblack+\nowwhite\leq n-2\sqrt{n\log n}$.}
    From \eqref{eq:boost_before_cases-1}, we have
    \begin{align}
        \E\left[\left(\nextblack+\nextwhite\right)\right]
        %& \geq\nowblack+\nowwhite+\left(n-\nowblack-\nowwhite\right)
        %    \left(1-\exp\left(-\frac{\nowblack+\nowwhite}{n}\newconstboost\right)\right)\\
        & \geq\nowblack+\nowwhite+\left(n-\nowblack-\nowwhite\right)\left(1-\frac{1}{\sqrt{e}}\right)
        \geq n\left(1-\frac{1}{\sqrt{e}}\right)+\frac{1}{\sqrt{e}}\left(\nowblack+\nowwhite\right).
    \end{align}
    From the Chernoff bound (Lemma \ref{thm:gen_cb}), we thus get that
    w.h.p. 
    \[
    \nextblack+\nextwhite\geq n\left(1-\frac{1}{\sqrt{e}}\right)+\frac{1}{\sqrt{e}}\left(\nowblack+\nowwhite\right)-\sqrt{n\log n}.
    \]

    \textbf{Case $\nowblack+\nowwhite>n-2\sqrt{n\log n}$.} The probability
    that a silent agent does not observe a speaking one is 
    \[
        \left(\frac{n-\nowblack-\nowwhite}{n}\right)^{\newconstboost}\leq \left(\frac{4\log n}{n}\right)^{\frac{1}{2}\newconstboost},
    \]
    hence by a simple union bound it follows that w.h.p.
    all agents are speaking. 

    Now, we prove \eqref{eq:rel_bias}. 
    As in the proof of \eqref{eq:boost_bias}, we have two cases. 
    The first case, $\frac{\sizeblack}{\sizewhite}\geq\frac{n}{2\newconstboost}$,
    is a simple consequence of the Chernoff bound (Lemma \ref{thm:gen_cb}). 
    Otherwise, let us assume $\frac{\sizeblack}{\sizewhite}<\frac{n}{2\newconstboost}$.
    With an analogous argument to that
    for \eqref{eq:exp_white_boost} and \eqref{eq:exp-ratio} we can prove 
    \begin{equation}
        \frac{\E\left[\nextblack\right]}{\E\left[\nextwhite\right]}  =\frac{\nowblack}{\nowwhite},\label{eq:expec_new_ratio}
    \end{equation}
    and
    \begin{align}
        & \E\left[\nextblack\right] 
        = \nowblack + \frac{\nowblack}{\nowblack+\nowblack}
             \E\left[\nextblack-\nowblack+\nextwhite-\nowwhite\right],\\
        & \E\left[\nextwhite\right]
        = \nowwhite + \frac{\nowwhite}{\nowblack+\nowwhite}
            \E\left[\nextblack-\nowblack+\nextwhite-\nowwhite\right].
        \label{eq:min_exp_next_white}
    \end{align}
     As in \eqref{eq:exp_white_boost_whp}, from the multiplicative form
    of the Chernoff bound (Corollary \ref{cor:cb}) we have that w.h.p.
    \begin{align}
        & \nextblack  \geq\E\left[\nextblack\right]-\sqrt{\E\left[\nextblack\right]\log n}\quad\mbox{and}\\
        & \nextwhite  \leq\E\left[\nextwhite\right]+\sqrt{\E\left[\nextwhite\right]\log n}.\label{eq:cb_newratios}
    \end{align}

    Thus, as in \eqref{eq:recurs_whp}, from \eqref{eq:cb_newratios} and \eqref{eq:expec_new_ratio}, we get
    \begin{align}
        \frac{\nextblack}{\nextwhite}
        & \geq\frac{\E\left[\nextblack\right]}{\E\left[\nextwhite\right]}
            \cdot \left(1-2\sqrt{\frac{\log n}{\E\left[\nextwhite\right]}}\right)
        \geq\frac{\nowblack}{\nowwhite} 
            \cdot \left(1- 4\sqrt{\frac{\log n}{\newconstboost\nowwhite}} \right),
    \end{align}
    where, as in \eqref{eq:exp_white_boost_lower}, in the last inequality we used that from \eqref{eq:exp_spread-1} and \eqref{eq:min_exp_next_white} it holds $ \E\left[\nextwhite\right] \geq \frac \newconstboost4 \nowwhite$. 
\end{proof}

From the previous two lemmas, we can derive the following corollary,
which concludes the proof. 
\begin{corollary}\label{cor:ultimate}
    If $\sizeblack \geq \sizewhite (1+\epsilon)$ for some constant
    $\epsilon>0$, then at the end of the last spreading phase it holds w.h.p.
    \begin{align}
        \sizeblack^{\left(1+2\log n\right)} 
        & =n-\sizewhite^{\left(1+2\log n\right)}
        \geq\sizewhite^{\left(1+2\log n\right)}\left(1+\trialbias\right),
        \label{eq:full_opinion}
    \end{align}
    where 
    $\trialbias= \frac {\epsilon} 2 -\frac{4}{\sqrt{\newconstboost}}$.
\end{corollary}
\begin{proof}
    We first show how the equality in \eqref{eq:full_opinion} follows from \eqref{eq:increase_big}. 
    When $\nowblack+\nowwhite<\frac{n}{2\constphase}$, equation 
    \eqref{eq:increase_big} shows that $\nowblack+\nowwhite$ increases by multiplicative a factor 
    $\constphase$ at the end of each spreading phase. 
    When $\frac{n}{2\constphase} \leq \nowblack+\nowwhite \leq n-2\sqrt{n\log n}$
    equation \eqref{eq:increase_big} shows that
    \begin{align}
        n- \nextblack -\nextwhite 
        & \leq \frac{ n-\nowblack-\nowwhite }{\sqrt{e}} - \sqrt{n \log n} 
        \leq \frac{n-\nowblack-\nowwhite }{\sqrt{e}}.
    \end{align}
    Hence the number of silent agents decreases by a factor $\sqrt{e}$ after each spreading phase.
    Lastly, when
    $\nowblack+\nowwhite>n-2\sqrt{n\log n}$, after one more spreading phase, a simple application of the  union bound shows that $\nextblack+\nextwhite$ is equal to $n$ w.h.p. As a consequence,
    if $\newconstboost$ is big enough, %
    after less than
    $ 1+ 2\log n $
    spreading phases w.h.p it holds that
    $
    \sizeblack^{\left(1+2\log n\right)} =n-\sizewhite^{\left(1+2\log n\right)}
    $. 

    The inequality in \eqref{eq:full_opinion} can be derived from \eqref{eq:rel_bias}, as
    follows. From \eqref{eq:boost_bias} and \eqref{eq:rel_bias} we have
    \begin{align}
        & \frac{\sizeblack^{\left(1+2\log n\right)}}{\sizewhite^{\left(1+2\log n\right)}} 
        \geq\frac{\sizeblack}{\sizewhite}
            \left(1-\sqrt{\frac{9}{\newconstboost\sizewhite}}\right)
            \prod_{i=2}^{1+2\log n}\left(1-\sqrt{\frac{16\log n}{\newconstboost\nowwhite}}\right).
        \label{eq:final_ratio}
    \end{align}
    We can estimate the product as
    \begin{align}
         \prod_{i=2}^{1+2\log n}\left(1-\sqrt{\frac{16\log n}{\newconstboost\nowwhite}}\right)
        &\geq\exp\left(- 4 \sum_{i=2}^{1+2\log n}\frac{1}
            {\left(\sqrt{\newconstboost}\right)^{i}}\right),\\
        & \geq\exp\left\{ 4 \left(1+\frac{1}{\sqrt{\newconstboost}}
            -\frac{1-\left(\newconstboost\right)^{-\frac{2+2\log n}{2}}}
                {1-\left(\newconstboost\right)^{-\frac{1}{2}}}\right)\right\}\\
        & \geq\exp\left\{- 4 \left(\frac{1}{\newconstboost-\sqrt{\newconstboost}}
            -n^{-\frac{2\log\newconstboost}{2}}\right)\right\}\\
        & \geq\left(1-\frac{5}{\newconstboost}\right),
        \label{eq:prod_lower}
    \end{align}
    where in the first and last inequality we used that $1-x\geq e^{-\frac{x}{1-x}}$
    if $\left|x\right|<1$.

    Finally, from \eqref{eq:final_ratio} and \eqref{eq:prod_lower} we
    get
    \begin{align}
        \frac{\sizeblack^{\left(1+2\log n\right)}}{\sizewhite^{\left(1+2\log n\right)}}
        & \geq\frac{\sizeblack}{\sizewhite}
            \left(1-\sqrt{\frac{9}{\newconstboost\sizewhite}}\right)
            \left(1-\frac{5}{\newconstboost}\right)
        \geq\frac{\sizeblack}{\sizewhite}\left(1-\frac{4}{\sqrt{\newconstboost}}\right),
    \end{align}
    which, together with the hypothesis 
    $\frac{\sizeblack}{\sizewhite} \geq 1+ \epsilon$,
    concludes the proof.
\end{proof}

\subsection{Proof of Theorem \ref{thm:multi}}\label{sub:proofofmulti}

\begin{proof}
    From Corollary \ref{cor:ultimate}, it follows that at the end of the last spreading phase, all agents have been informed. 
    After the last spreading phase, during the polling phase, each agent samples $\newconstboost \log n$ opinions from the \population and then adopts the majority of these as her output bit. 
    Thus, \eqref{eq:full_opinion} ensures that
    each sample 
    holds the correct opinion with probability
    $ \geq \frac{1}{2}+ \trialbias$. Hence, by the Chernoff bound and a union bound, if $\newconstboost$ is big enough then the majority of the $\newconstboost \log n$ samples corresponds to the correct value for all the $n$ agents w.h.p.

    The protocol obtained so far solves \majority \broadcast, but it does it using $4$ bits per message rather than $3$. Indeed, synchronizing a clock using \synclock takes $3$ bits, and we use an extra bit to execute \majprotgiven described in Section \ref{sec:majprot}.
    However, the protocol \majprot has the \independence property. This follows from Lemma \ref{lem:ipfact} with \prota$=$\synclock, \protb$=$\majprotgiven, \protc$=$\majprot, together with the observation that \majprotgiven is self-stabilizing.
    We can thus reduce the message length of \majprot  to $3$ bits using again \corelemma, with a time overhead of a factor $4$ only.
\end{proof}

\section{Conclusion and Open Problems}

This paper deals with  the construction of protocols in highly congested
stochastic interaction patterns. Corresponding challenges are particularly
evident when it is difficult to guarantee synchronization, which seems to be
essential for emulating a typical protocol that relies on many bits per message
with a protocol that uses fewer bits. Our paper shows that in the $\pull$
model, if a self-stabilizing protocol satisfies the \independence property then
it can be emulated with only 3 bits per message. Using this rather general
transformer, we solve the self-stabilizing  Clock-Synchronization and \majority
\broadcast problems in almost-logarithmic time and using only $3$ bits per
message. It remains an open problem whether the message size of either one of
these problems can be further reduced while keeping the running time
polylogarithmic. 

In particular, even for the self-stabilizing \broadcast problem (with a single
source) it remains open whether there exists a polylogarithmic protocol that
uses a single bit per interaction. 
In fact, we investigated several candidate protocols which seem promising in
experimental simulation, but appear to be out of reach of current techniques
for analysing randomly-interacting agent systems in a self-stabilizing context. 
Let us informally present one of them, called \textsc{BFS}\footnote{A similar
protocol was suggested during discussions with Bernhard Haeupler.}.
Let $\ell$,$ k\in \NN$ be two parameters, say of order $O(\log n)$. Agents can
be in $3$ states: \emph{boosting}, \emph{frozen} or \emph{sensitive}. 
Boosting agents behave as in the \major protocol: they apply the majority rule
to the $2$ values they see in a given round and make it into their opinion for
the next round. They also keep a counter $T$. If they have seen only agents of
a given color $b$ for $\ell $ rounds, they become sensitive to the opposite
value.
$b$-sensitive agents turn into frozen-$b$ agents if they see value $b$.
$b$-frozen agents keep the value $b$ for $k$ rounds before becoming boosters again.
Intuitively what we expect is that, from every configuration, at some point
almost all agents would be in the {boosting} state. Then, the boosting behavior
would lead the agents to converge to a value $b$ (which depends on the initial
conditions). Most agents would then become sensitive to $1-b$. If the source
has opinion $1-b$ then there should be a ``switch'' from $b$ to $1-b$. The
``frozen'' period is meant to allow for some delay in the times at which agents
become sensitive, and then flip their opinion. 
%This algorithm however, as the other candidate protocols we have considered,
%does not seem amenable to a rigorous analysis in the $\pull$ model which
%accounts for the self-stabilizing requirement. 

\medskip
{ {\paragraph{Acknowledgments:} The problem of self-stabilizing \broadcast was introduced through discussions with Ofer Feinerman. The authors are also thankful for Omer Angel, Bernhard Haeupler, Parag Chordia, Iordanis Kerenidis,  
Fabian Kuhn, Uri Feige, and Uri Zwick for helpful discussions regarding that problem.  The authors also thank Michele Borassi for his helpful suggestions regarding the Clock Synchronization problem.
}}

\bibliographystyle{abbrv}
\bibliography{biblio}

\begin{thebibliography}{10}

\bibitem{beeping1}
Y.~Afek, N.~Alon, O.~Barad, E.~Hornstein, N.~Barkai, and Z.~Bar-joseph.
\newblock A biological solution to a fundamental distributed computing problem.
\newblock {\em Science}, 2011.

\bibitem{Dan2}
D.~Alistarh and R.~Gelashvili.
\newblock Polylogarithmic-time leader election in population protocols.
\newblock In {\em {ICALP}}, pages 479--491, 2015.

\bibitem{Pop1}
D.~Angluin, J.~Aspnes, Z.~Diamadi, M.~J. Fischer, and R.~Peralta.
\newblock Computation in networks of passively mobile finite-state sensors.
\newblock {\em Distributed Computing}, 18(4):235--253, 2006.

\bibitem{Aspnes}
D.~Angluin, J.~Aspnes, and D.~Eisenstat.
\newblock A simple population protocol for fast robust approximate majority.
\newblock {\em Distributed Computing}, 21(2):87--102, 2008.

\bibitem{Pop-ss}
D.~Angluin, J.~Aspnes, M.~J. Fischer, and H.~Jiang.
\newblock Self-stabilizing population protocols.
\newblock {\em {TAAS}}, 3(4), 2008.

\bibitem{AFJ06}
D.~Angluin, M.~J. Fischer, and H.~Jiang.
\newblock {\em Stabilizing Consensus in Mobile Networks}, pages 37--50.
\newblock Springer, 2006.

\bibitem{Pop2}
J.~Aspnes and E.~Ruppert.
\newblock An introduction to population protocols.
\newblock {\em Bulletin of the {EATCS}}, 93:98--117, 2007.

\bibitem{Attiya}
H.~Attiya, A.~Herzberg, and S.~Rajsbaum.
\newblock Optimal clock synchronization under different delay assumptions.
\newblock {\em {SIAM} J. Comput.}, 25(2):369--389, 1996.

\bibitem{Joffroy}
J.~Beauquier, J.~Burman, and S.~Kutten.
\newblock A self-stabilizing transformer for population protocols with
  covering.
\newblock {\em Theor. Comput. Sci.}, 412(33):4247--4259, 2011.

\bibitem{BCN15b}
L.~Becchetti, A.~Clementi, E.~Natale, F.~Pasquale, and G.~Posta.
\newblock Self-stabilizing repeated balls-into-bins.
\newblock In {\em {SPAA}}, pages 332--339, 2015.

\bibitem{BCN15}
L.~Becchetti, A.~E.~F. Clementi, E.~Natale, F.~Pasquale, and R.~Silvestri.
\newblock Plurality consensus in the gossip model.
\newblock In {\em SODA}, pages 371--390, 2015.

\bibitem{BCN16}
L.~Becchetti, A.~E.~F. Clementi, E.~Natale, F.~Pasquale, and L.~Trevisan.
\newblock Stabilizing consensus with many opinions.
\newblock In {\em {SODA}}, pages 620--635, 2016.

\bibitem{Ben-Or}
M.~Ben{-}Or, D.~Dolev, and E.~N. Hoch.
\newblock Fast self-stabilizing byzantine tolerant digital clock
  synchronization.
\newblock In {\em {PODC}}, pages 385--394, 2008.

\bibitem{BKN16}
L.~Boczkowski, A.~Korman, and E.~Natale.
\newblock Brief announcement: Self-stabilizing clock synchronization with 3-bit
  messages.
\newblock In {\em {PODC}}, 2016.

\bibitem{BKN17}
L.~Boczkowski, A.~Korman, and E.~Natale.
\newblock Minimizing message size in stochastic communication patterns: Fast
  self-stabilizing protocols with 3~bits.
\newblock In {\em {SODA}}, 2017.

\bibitem{CHHKM12}
K.~Censor{-}Hillel, B.~Haeupler, J.~A. Kelner, and P.~Maymounkov.
\newblock Global computation in a poorly connected world: fast rumor spreading
  with no dependence on conductance.
\newblock In {\em {STOC}}, pages 961--970, 2012.

\bibitem{CCDS14}
H.-L. Chen, R.~Cummings, D.~Doty, and D.~Soloveichik.
\newblock Speed faults in computation by chemical reaction networks.
\newblock {\em In Distributed Computing}, pages 16--30, 2014.

\bibitem{CLP09}
F.~Chierichetti, S.~Lattanzi, and A.~Panconesi.
\newblock Rumor spreading in social networks.
\newblock In {\em ICALP}, pages 375--386, 2009.

\bibitem{CER15}
C.~Cooper, R.~Els{\"a}sser, T.~Radzik, N.~Rivera, and T.~Shiraga.
\newblock Fast consensus for voting on general expander graphs.
\newblock In {\em DISC}, pages 248--262. Springer, 2015.

\bibitem{COU05}
I.~Couzin, J.~Krause, N.~Franks, and S.~Levin.
\newblock Effective leadership and decision making in animal groups on the
  move.
\newblock {\em Nature 433}, pages 513--516, 2005.

\bibitem{FishConsensus}
e.~a. D.~JT.~Sumpter.
\newblock Consensus decision making by fish.
\newblock {\em Current Biology 22(25)}, pages 1773--1777, 2008.

\bibitem{DGH88}
A.~J. Demers, D.~H. Greene, C.~Hauser, W.~Irish, J.~Larson, S.~Shenker, H.~E.
  Sturgis, D.~C. Swinehart, and D.~B. Terry.
\newblock Epidemic algorithms for replicated database maintenance.
\newblock {\em Operating Systems Review}, 22(1):8--32, 1988.

\bibitem{dijkstra}
E.~W. Dijkstra.
\newblock Self-stabilizing systems in spite of distributed control.
\newblock {\em Commun. {ACM}}, 17(11):643--644, 1974.

\bibitem{DF11a}
B.~Doerr and M.~Fouz.
\newblock Asymptotically optimal randomized rumor spreading.
\newblock {\em Electronic Notes in Discrete Mathematics}, 38:297--302, 2011.

\bibitem{DGM11}
B.~Doerr, L.~A. Goldberg, L.~Minder, T.~Sauerwald, and C.~Scheideler.
\newblock Stabilizing consensus with the power of two choices.
\newblock In {\em {SPAA}}, pages 149--158, 2011.

\bibitem{Dolev07}
D.~Dolev and E.~N. Hoch.
\newblock On self-stabilizing synchronous actions despite byzantine attacks.
\newblock In {\em {DISC}}, pages 193--207, 2007.

\bibitem{DolevKLRS13}
D.~Dolev, J.~H. Korhonen, C.~Lenzen, J.~Rybicki, and J.~Suomela.
\newblock Synchronous counting and computational algorithm design.
\newblock In {\em SSS}, pages 237--250, 2013.

\bibitem{Dolev97}
S.~Dolev.
\newblock Possible and impossible self-stabilizing digital clock
  synchronization in general graphs.
\newblock {\em Real-Time Systems}, 12(1):95--107, 1997.

\bibitem{Dolev04}
S.~Dolev and J.~L. Welch.
\newblock Self-stabilizing clock synchronization in the presence of byzantine
  faults.
\newblock {\em J. {ACM}}, 51(5):780--799, 2004.

\bibitem{Doty}
D.~Doty and D.~Soloveichik.
\newblock Stable leader election in population protocols requires linear time.
\newblock {\em CoRR}, abs/1502.04246, 2015.

\bibitem{EFK16}
R.~Els{\"{a}}sser, T.~Friedetzky, D.~Kaaser, F.~Mallmann{-}Trenn, and
  H.~Trinker.
\newblock Efficient k-party voting with two choices.
\newblock {\em CoRR}, abs/1602.04667, 2016.

\bibitem{ES09}
R.~Els{\"{a}}sser and T.~Sauerwald.
\newblock On the runtime and robustness of randomized broadcasting.
\newblock {\em Theor. Comput. Sci.}, 410(36):3414--3427, 2009.

\bibitem{emek}
Y.~Emek and R.~Wattenhofer.
\newblock Stone age distributed computing.
\newblock In {\em PODC}, pages 137--146, 2013.

\bibitem{OHK14}
O.~Feinerman, B.~Haeupler, and A.~Korman.
\newblock Breathe before speaking: efficient information dissemination despite
  noisy, limited and anonymous communication.
\newblock In {\em {PODC}}, pages 114--123, 2014.

\bibitem{SIROCCO}
O.~Feinerman and A.~Korman.
\newblock Clock synchronization and estimation in highly dynamic networks: An
  information theoretic approach.
\newblock In {\em {SIROCCO}}, pages 16--30, 2015.

\bibitem{Survey}
O.~Feinerman and A.~Korman.
\newblock Individual versus collective cognition in social insects.
\newblock {\em Submitted to Journal of Experimental Biology}, 2016.

\bibitem{HM85}
R.~Harkness and N.~Maroudas.
\newblock Central place foraging by an ant ({\em cataglyphis bicolor} fab.): a
  model of searching.
\newblock {\em Animal Behavior 33(3)}, pages 916--928, 1985.

\bibitem{Herman}
T.~Herman.
\newblock Phase clocks for transient fault repair.
\newblock {\em {IEEE} Trans. Parallel Distrib. Syst.}, 11(10):1048--1057, 2000.

\bibitem{KSSV00}
R.~M. Karp, C.~Schindelhauer, S.~Shenker, and B.~V{\"{o}}cking.
\newblock Randomized rumor spreading.
\newblock In {\em {FOCS}}, pages 565--574, 2000.

\bibitem{KDG03}
D.~Kempe, A.~Dobra, and J.~Gehrke.
\newblock Gossip-based computation of aggregate information.
\newblock In {\em FOCS}, pages 482--491. IEEE, 2003.

\bibitem{kutten}
A.~Kravchik and S.~Kutten.
\newblock Time optimal synchronous self stabilizing spanning tree.
\newblock In Y.~Afek, editor, {\em {DISC}, Jerusalem, Israel, October 14-18,
  2013. Proceedings}, volume 8205 of {\em Lecture Notes in Computer Science},
  pages 91--105. Springer, 2013.

\bibitem{Lamport}
L.~Lamport.
\newblock Time, clocks, and the ordering of events in a distributed system.
\newblock {\em Commun. {ACM}}, 21(7):558--565, 1978.

\bibitem{Lenzen3}
C.~Lenzen, T.~Locher, P.~Sommer, and R.~Wattenhofer.
\newblock Clock synchronization: Open problems in theory and practice.
\newblock In {\em SOFSEM}, pages 61--70, 2010.

\bibitem{Lenzen2}
C.~Lenzen, T.~Locher, and R.~Wattenhofer.
\newblock Tight bounds for clock synchronization.
\newblock {\em J. {ACM}}, 57(2), 2010.

\bibitem{Lenzen5}
C.~Lenzen and J.~Rybicki.
\newblock Efficient counting with optimal resilience.
\newblock In {\em DISC}, pages 16--30, 2015.

\bibitem{Lenzen4}
C.~Lenzen, J.~Rybicki, and J.~Suomela.
\newblock Towards optimal synchronous counting.
\newblock In {\em PODC}, pages 441--450, 2015.

\bibitem{MCD98}
C.~McDiarmid.
\newblock {\em Concentration}, pages 195--248.
\newblock Springer, 1998.

\bibitem{Razin}
N.~Razin, J.~Eckmann, and O.~Feinerman.
\newblock Desert ants achieve reliable recruitment across noisy interactions.
\newblock {\em Journal of the Royal Society Interface; 10(20170079).}, 2013.

\bibitem{ManyEyes}
G.~Roberts.
\newblock Why individual vigilance increases as group size increases.
\newblock {\em Animal Behaviour 51}, pages 1077--1086, 1996.

\end{thebibliography}

\appendix

\section{Technical Tools}
\label{apx:tech}

\begin{theorem}[\cite{MCD98}]
    \label{thm:gen_cb}
    Let $X_{1},...,X_{n}$ be $n$ independent random
    variables. If $X_{i}\leq M$ for each $i$, then
    \begin{align}
        ~~~~\Pr\left(\sum_{i}X_{i}\geq\E\left[\sum_{i}X_{i}\right]+\lambda\right)
        \leq e^{ -\frac{\lambda^{2}}
            {2\left(\sqrt{\sum_{i}\E\left[X_{i}^{2}\right]}+\frac{M\lambda}{3}\right)}}.
        \label{eq:gen_cb_upper}
    \end{align}
\end{theorem}
\begin{corollary}
    \label{cor:cb} Let $\mu = \E\left[\sum_{i}X_{i}\right]$. If the $X_{i}$s are binary then, 
    for $\lambda=\sqrt{ \mu \log n}$ and sufficiently large $n$, 
    \eqref{eq:gen_cb_upper} gives
    \begin{align}
        \Pr\left(\sum_{i}X_{i}\geq \mu +\sqrt{\mu \log n}\right) & \leq e^{-\sqrt{ \mu \log n}},\\
        \Pr\left(\sum_{i}X_{i}\leq \mu -\sqrt{ \mu \log n}\right) & \leq e^{-\sqrt{ \mu \log n}}.
    \end{align}
\end{corollary}
\begin{proof}
    The fact that the $X_{i}$s are binary implies that
    $\sum_{i}\E\left[X_{i}^{2}\right]\leq\sum_{i}\E\left[X_{i}\right]$.
    By setting $\lambda=\sqrt{\E\left[\sum_{i}X_{i}\right]\log n}$, one can
    show that the l.h.s. of \eqref{eq:gen_cb_upper} is upper bounded by
    $e^{-\sqrt{ \mu \log n}}$. 
\end{proof}

\section{Proof of Lemma \ref{lem:gt}}
\begin{lemma}\label{lem:gt}
    Let $f,g: \RR_{+} \rightarrow \RR$ be functions defined by $f(x) =  \lceil \log x \rceil +1$ and 
    \begin{align}
        \finaliter(x) = \inf \left\{k \in \NN \mid f^{\circledast k}(x) \leq 3 \right\},
    \end{align}
    where we denote by $f^{\circledast k}$ the $k$-fold iteration of $f$.
    It holds that
    \begin{align}
        \finaliter(T) \leq \log^{\circledast 4 }T + \bigO(1).
    \end{align}
\end{lemma}
\begin{proof}
    We can notice that
    $ 
        f(T) \leq T -1, 
    $
    if $T$ is bigger than some constant $c$.  Moreover, when $f(x) \leq c$, the
    number of iterations before reaching $1$ is $O(1)$. This implies that
    $
        \finaliter(T) \leq T + \bigO(1).
    $
    But in fact, by definition, $\ell(T) = g\left(f^{\circledast 4}(T)\right) +4$ (provided $f^{\circledast 4}(T) > 1$, which holds if $T$ is big enough). Hence
    \begin{align}
        \finaliter(T) 
        \leq g\left(f^{\circledast 4}(T)\right) +4
        \leq f^{\circledast 4}(T) + \bigO(1) \leq \log^{\circledast 4} T + \bigO(1).
    \end{align}
\end{proof}

\end{document}